\documentclass[twocolumn,10pt]{IEEEtran}
\usepackage{braket}

\usepackage[normalem]{ulem}
\usepackage{url}
\usepackage{hyperref}
\hypersetup{colorlinks,urlcolor=blue}

\hypersetup{colorlinks=false,pdfborder=000}
\input{def.tex}
\usepackage{dsfont}
\usepackage{minibox}
\DeclareSymbolFont{matha}{OML}{txmi}{m}{it}
\DeclareMathSymbol{\varv}{\mathord}{matha}{118}
\usepackage{eurosym}
\usepackage{hhline,physics}

\usepackage{multicol}
\usepackage{algorithm}
\usepackage{graphicx}
\usepackage{textcomp}
\usepackage{caption,pifont}
\usepackage[multiple]{footmisc}

\usepackage{algorithmic}

\makeatletter

\IEEEoverridecommandlockouts
\begin{document}
	\title{Performance of Double-Stacked Intelligent Metasurface-Assisted Multiuser Massive MIMO Communications in the Wave Domain } 
	\author{Anastasios Papazafeiropoulos, Pandelis Kourtessis, Symeon Chatzinotas, Dimitra I. Kaklamani, 			Iakovos S. Venieris \thanks{A. Papazafeiropoulos is with the Communications and Intelligent Systems Research Group, University of Hertfordshire, Hatfield AL10 9AB, U. K., and with SnT at the University of Luxembourg, Luxembourg.  P. Kourtessis is with the Communications and Intelligent Systems Research Group, University of Hertfordshire, Hatfield AL10 9AB, U. K.  S. Chatzinotas is with the SnT at the University of Luxembourg, Luxembourg. Dimitra I. Kaklamani is with the Microwave and Fiber Optics Laboratory, and Iakovos S. Venieris is  with the Intelligent Communications and Broadband Networks Laboratory, School of Electrical and Computer Engineering, National Technical University of Athens, Zografou, 15780 Athens,	Greece.	
			Corresponding author's email: tapapazaf@gmail.com.}}
	\maketitle\vspace{-1.7cm}
	\begin{abstract}	
		Although reconfigurable intelligent surface (RIS) is a promising technology for shaping the propagation environment, it consists of a single-layer structure within inherent limitations regarding the number of beam steering patterns. Based on the recently revolutionary technology, denoted as stacked intelligent metasurface (SIM), we propose its implementation not only on the base station (BS) side in a massive multiple-input multiple-output (mMIMO) setup but also in the intermediate space between the base station and the users to adjust the environment further as needed. For the sake of convenience, we call the former BS SIM (BSIM), and the latter channel SIM (CSIM). To this end, we achieve \textcolor{black}{hybrid} wave-based combining at the BS and wave-based  configuration at the intermediate space.   Specifically, we propose a channel estimation method with reduced overhead, being crucial for SIM-assisted communications. Next, we derive the uplink sum spectral efficiency (SE) in closed form in terms of statistical channel state information (CSI). Notably, we optimize the phase shifts of both BSIM and CSIM simultaneously by using the   projected  gradient ascent method (PGAM). Compared to previous works on SIMs, we study the uplink transmission in a mMIMO setup, channel estimation in a single phase, a second SIM at the intermediate space, and simultaneous optimization of the two SIMs. Simulation results show the impact of various parameters on the sum SE, and demonstrate the superiority of our optimization approach compared to the alternating optimization (AO) method.
	\end{abstract}
	\begin{keywords}
		Reconfigurable intelligent surface 	(RIS), stacked intelligent metasurfaces (SIM),  gradient projection,  6G networks.
	\end{keywords}
	
	\section{Introduction}
	In the pursuit of meeting the stringent demands of future wireless networks systems \cite{Boccardi2014}, massive multiple-input multiple-output (mMIMO) systems and millimeter-wave (mmWave) communications, suggested the previous years, have failed to meet the low energy consumption and the ubiquitous wireless connectivity requirements \cite{Zhang2020b}. For example, mMIMO systems exhibit low performance in poor scattering conditions despite their implementation with a large number of active elements which may also consume excessive energy \cite{Sohrabi2016,Zhang2020b}. In the case of mmWave communications, these require costly and power-hungry transceivers that can be easily blocked by  common obstacles such as walls \cite{Rappaport2015}. It is thus evident that it is crucial for next-generation networks to introduce a radical  paradigm,  which will give priority to energy sustainability  and  pervasive connectivity through the emergence of programmable wireless environments.
	
	In this direction, a prominent technology that can dynamically control the wireless environment in an almost passive manner while enhancing significantly the wireless channel quality is the reconfigurable intelligent surface (RIS) \cite{Wu2019,Basar2019,Bjoernson2020,Papazafeiropoulos2021,Papazafeiropoulos2022}. Generally, an RIS consists of an artificial surface implemented by a large number of cost-efficient nearly passive elements, which can induce independently phase shifts on the impinging electromagnetic (EM) waves through a smart controller. Thus, the RIS is capable of shaping dynamically the  propagation environment and reducing the energy consumption \cite{Huang2019}.
	
	Despite that RIS has been suggested for various communication scenarios because of its numerous advantages \cite{Wu2019,Basar2019,Bjoernson2020,DiRenzo2020,Kammoun2020,Yang2020b,Pan2020,Han2019,Zhang2021a,Zhao2020, Papazafeiropoulos2021,Papazafeiropoulos2021a,Gan2022,Zhi2023,Papazafeiropoulos2023,Papazafeiropoulos2023a}, existing research works have mostly considered single-layer metasurface structures, which limit the degrees of freedom for managing the beams. Moreover, it has been shown that RIS cannot suppress the inter-user interference because of their single-layer configuration and hardware limitations \cite{An2023}. The identification of these gaps made the authors in \cite{An2023} propose a stacked intelligent metasurface (SIM)  with remarkable signal processing capabilities compared to conventional RIS with a single layer. The proposed SIM is not a pure mathematical abstraction but it is based on tangible hardware prototypes \cite{Lin2018,Liu2022} and the technological  advancements in wave-based computing. In particular, \cite{Lin2018} proposed a deep neural network ($ \mathrm{D}^{2}\mathrm{NN} $), which has a similar structure to a SIM to perform parallel calculations at the speed of light by employing three-dimensional (3D) printed optical lenses and taking advantage of the wave properties of photons. However, this fabrication is fixed and cannot be retrained to perform other tasks. To this end, in \cite{Liu2022}, a programmable $ \mathrm{D}^{2}\mathrm{NN} $, which is closer to a SIM was designed with each reprogrammable artificial  neuron acting as the SIM meta-atoms. Hence,  they demonstrated the execution of various tasks, e.g., image classification through the flexible manipulation of the propagation of the EM waves through its numerous layers.
	
	These experimental results led to propose a SIM-based transceiver for point-to-point MIMO communications \cite{An2023}, where two SIMs are implemented at the  transmitter and the receiver, respectively  to perform precoding and  combining without any digital hardware, which reduced the need for a large number of antennas while the EM signals propagate through them. Therein, the phase shifts of each  SIM were optimised in an alternating optimisation (AO) manner based on instantaneous  channel state information (CSI) conditions. In \cite{Papazafeiropoulos2023}, we proposed a simultaneous optimization of the phase shifts of both the transmitter and the  receiver.   In \cite{An2023b},  the authors deployed a SIM-enabled base station (BS) to perform downlink beamforming  in the EM wave domain to multiple users. However, these works relied on perfect CSI  to design the SIMs. In \cite{Nadeem2023}, a hybrid wave-based  channel estimator was proposed with multiple subphases, where symbols were first precoded in the wave domain, and then, in the digital domain.
	
	\subsection{Contributions}
	\textcolor{black}{The reflections above led us to the crux of this work,  which is to introduce two types of SIMs in a mMIMO system with multiple users, one hybrid SIM at the BS and one at the  intermediate space under the assumptions of correlated fading and imperfect CSI and maximum-ratio combining (MRC), and study the uplink sum  spectral efficiency (SE). To this end, we apply the two-timescale methodology \cite{Zhao2020, Papazafeiropoulos2021,Zhi2023}, where the beamforming is based on instantaneous CSI and the optimization is based on statistical CSI to save significant overhead.}

	\textcolor{black}{	Our main contributions are summarised as follows:}
	\begin{itemize}
		\item 
		\textcolor{black}{	The two types of SIMs increase the technical novelty. 
			\begin{itemize}
				\item 		 \textcolor{black}{Notably, one of the main novelties of this work is the use  of  a hybrid digital and wave-based  SIM at the BS instead of a fully wave-based  SIM  to reduce the metasurface size appropriately since a fully wave-domain architecture  requires a large number of meta-atoms. It has to be mentioned that the advantage with respect to  a wave-based SIM regarding the reduction of the number of RF chains at the BS is not affected significantly according to simulations.}		
				\item \textcolor{black}{The introduction of an environment SIM instead of a conventional RIS to shape the propagation environment provides more degrees of freedom to enable local multiplexing in the vicinity of the SIM with less interference since a SIM can produce a non-diagonal response matrix, while  the conventional RIS has a diagonal response matrix.}
		\end{itemize}}
		\item  \textcolor{black}{Different from \cite{An2023b}, which assumes a wave-based architecture with a SIM only at the BS, not only do we consider a hybrid SIM-based BS, but we also consider a SIM in the intermediate space to aid the  communication between the users and the BS. In other words, we consider two types of SIMs. Also, we rely on  statistical CSI instead of instantaneous-based analysis in \cite{An2023b}, and we account for channel estimation. Moreover, we consider imperfect CSI, which is of practical interest. Compared to \cite{Nadeem2023}, which considers a hybrid wave-digital architecture with a SIM only at the BS that studies only the channel estimation, it assumes several subfaces to estimate the channel in SIM-assisted multi-user systems,  we manage to update the estimated channel in a single phase. Thus,  the approach in \cite{Nadeem2023}, requiring too much overhead, reduces the SE, while  our approach  does not appear this significant disadvantage. Note that \cite{Nadeem2023} does not study the SE but focuses only on channel estimation. }
		\item  \textcolor{black}{We conceive the problem of optimizing the phase shifts of both SIMs. Despite the nonconvexity of this problem, we propose a projected gradient ancient method (PGAM),  which optimises the phase shifts of both SIMs simultaneously. This is a noteworthy contribution since \cite{An2023} optimized the phase shifts in an AO way.} 
		\item \textcolor{black}{Simulations show the superiority of the proposed architecture,  examine which SIM has the greatest impact, and shed light on the effect of various system parameters on the sum SE. In particular, we show the advantages of the hybrid SIM at the BS compared to the fully wave-based SIM and the conventional digital BS.}
	\end{itemize} 
	
	\textit{Paper Outline}: The remainder of this paper is organized as follows. Section~\ref{System} describes the system and channel models of a BSIM and  CSIM-assisted mMIMO system. Section~\ref{ChannelEstimation} elaborates on the CE and the uplink data transmission with the obtained uplink sum SE. In Section \ref{PSConfig}, we present the problem formulation and the simultaneous optimization of both BSIM and CSIM. The numerical results are discussed in Section~\ref{Numerical}.  Section~\ref{Conclusion} concludes the paper.
	
	\textit{Notation}: Vectors and matrices are described by boldface lower and upper case symbols, respectively. The notations $(\cdot)^\T$, $(\cdot)^\H$, and $\tr\!\left( {\cdot} \right)$ denote the transpose, Hermitian transpose, and trace operators, respectively.  The notation $\EE\left[\cdot\right]$ expresses     the expectation operator, and $\diag\left(\bA\right) $ describes a vector with elements equal to the  diagonal elements of $ \bA $.  The notation  $\diag\left(\bx\right) $ denotes a diagonal  matrix whose elements are $ \bx $, while  $\bb \sim \cC\cN{(\b0,\mathbf{\Sigma})}$ denotes a circularly symmetric complex Gaussian vector with zero mean and a  covariance matrix $\mathbf{\Sigma}$. In the case of complex-valued $\mathbf{x}$ and ${\mathbf{y}}$, we denote $\langle\mathbf{x},\mathbf{y}\rangle=2\Re{\mathbf{x}^{\H}\mathbf{y}}$.

	\section{System and Channel Models}\label{System}
	\begin{figure}[!h]
		\begin{center}
			\includegraphics[width=0.9\linewidth]{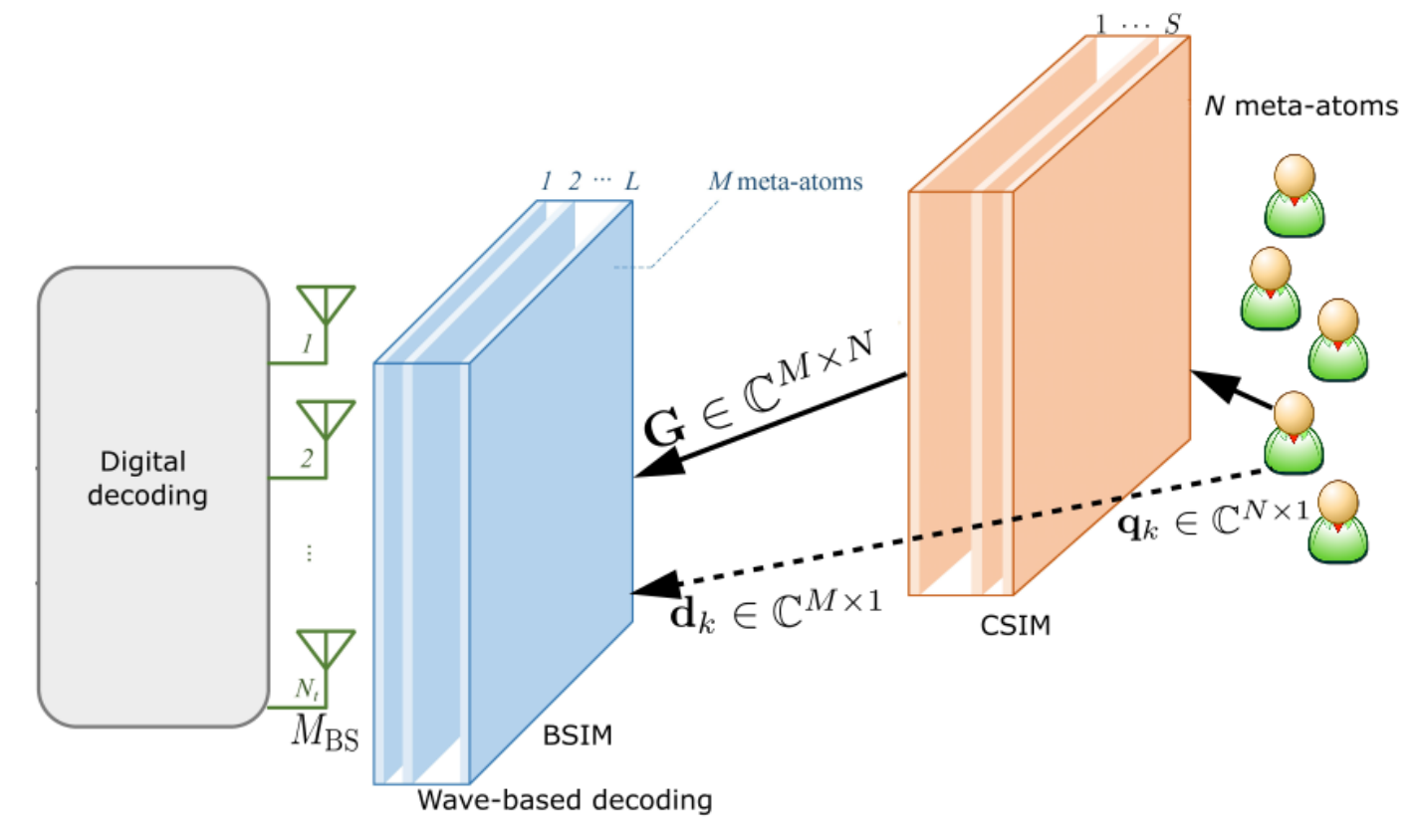}
			\caption{{ A BSIM and CSIM-assisted mMIMO system with multiple users.  }}
			\label{Fig1}
		\end{center}
	\end{figure} 
	\subsection{System Model}
	We consider the uplink of a multiuser MISO wireless system, where the BS is equipped with a large number of $ M_{\mathrm{BS}} $ antennas that serve $ K $ single-antennas users under the same time-frequency resources. In this context, we employ the BSIM, integrated into the BS to enable receive decoding in the EM wave domain, and the CSIM to enhance the communication between the users and the BS in the wave domain, as shown in Fig. \ref{Fig1}.
	For the fabrication of BSIM, we utilize an array of $ L $ metasurfaces, where each one of them consists of $ M $ meta-atoms. The corresponding sets are $ \mathcal{L}=\{1,\ldots,L\} $ and $\mathcal{M}=\{1,\ldots,M\} $. Similarly, the CSIM includes an array of $ S $ metasurfaces, each including $ N $  meta-atoms.\footnote{Without any loss of generality, we assume that the metasurfaces at each SIM include an identical number of  meta-atoms.} The corresponding sets are $ \mathcal{S}=\{1,\ldots,S\} $ and $\mathcal{N}=\{1,\ldots,N\} $. Both SIMs are connected to independent smart controllers, which can adjust  the corresponding  phase shifts of the EM waves that are transmitted through each of their meta-atoms. Note that the  forward propagation process in each SIM resembles a fully connected $ \mathrm{D}^{2}\mathrm{NN} $, where training the interconnection architecture of the  BSIM enables uplink decoding, while the respective training of the CSIM enables shaping the propagation environment. \textcolor{black}{It is worthwhile to mention that no special issue  appears in the case of a deployment of a CSIM since it is not significantly thicker than a conventional RIS. A typical number of layers is around $4$ as simulations show below.} We denote $ \bPhi^{l}=\diag(\bphi^{l})\in \mathbb{C}^{M \times M} $, where $ \bphi^{l} =[\phi^{l}_{1}, \dots, \phi^{l}_{M}]^{\T}\in \mathbb{C}^{M \times 1}$ is the  transmission coefficient matrix of the BSIM. Note that 
	$ \phi_{m}^{l} =e^{j \theta_{m}^{l}}$, where  $ \theta_{m}^{l}\in [0,2\pi), m \in \mathcal{M}, l \in \mathcal{L} $ is the phase shift by the $ m $-th meta-atom on the $ k $-th transmit metasurface layer. Similarly, for the CSIM, we denote $\bLambda^{s}=\diag(\blambda^{s})\in \mathbb{C}^{\textcolor{black}{N} \times N} $, where $ \blambda^{s} =[\lambda^{s}_{1}, \dots, \lambda^{s}_{M}]^{\T}\in \mathbb{C}^{N \times 1}$. Here, $ \lambda_{n}^{s}=e^{j \xi_{n}^{s}} $ with $ \xi_{n}^{s}\in [0,2\pi), n \in \mathcal{N}, s \in \mathcal{S} $ being the phase shift by the $ n $-th meta-atom on the $ s $-th CSIM metasurface layer.\footnote{We assume that the phase shifts are continuously-adjustable. Also, we assume that the modulus for the coefficient matrices of both SIMs equals $ 1 $ to evaluate the  maximization of the achievable rate, e.g., see \cite{Wu2019}. The study of more practical designs such as the consideration of  coupled phase and magnitude \cite{Abeywickrama2020} could be the topic of future work.}
	
	\subsection{Channel Model}
	The transmission coefficient from the $ \tilde{m} $th meta-atom on the $ (l-1) $st BSIM metasurface layer to the $ m $th meta-atom on the $ l $th BSIM metasurface layer, provided by the Rayleigh-Sommerfeld diffraction theory \cite{Lin2018}, is written as 
	\begin{align}
		w_{m,\tilde{m}}^{l}=\frac{A_{t}\cos x_{m,\tilde{m}}^{l}}{r_{m,\tilde{m}}^{l}}\left(\frac{1}{2\pi r^{l}_{m,\tilde{m}}}-j\frac{1}{\lambda}\right)e^{j 2 \pi r_{m,\tilde{m}}^{l}/\lambda}, l \in \mathcal{L},\label{deviationTransmitter}
	\end{align}
	where $ A_{t} $ expresses the area of each meta-atom, $ x_{m,\tilde{m}}^{l} $ is the angle between the propagation direction and the normal direction of the $ (l-1) $th BSIM metasurface layer, and $ r_{m,\tilde{m}}^{l} $, is the corresponding transmission distance. As a result, the effect of the BSIM can be modeled as
	\begin{align}
		\bP=\bPhi^{L}\bW^{L} \cdots \bPhi^{2}	\bW^{2} \bPhi^{1} \in \mathbb{C}^{M \times M },\label{TransmitterSIM}
	\end{align}
	where $ \bW^{l}\in \mathbb{C}^{M \times M}, l \in \mathcal{L}/\{1\} $ is the transmission coefficient matrix between the $ (l-1)st $ BSIM metasurface layer and the $ l $th transmit metasurface layer. Regarding $ \bW^{1} \in \mathbb{C}^{M \times M_{\mathrm{BS}}} $, it is the transmission coefficient matrix between the $ M_{\mathrm{BS}} $ transmit antennas
	and the first metasurface layer of the BSIM.

	In the CSIM,  the transmission coefficient from the $ n $th meta-atom
	on the $ s $th  metasurface layer to the $ \tilde{n} $th meta-atom on the $ (s-1) $st metasurface layer can be written as
	\begin{align}
		u_{\tilde{n},n}^{s}=\frac{A_{r}\cos \zeta_{\tilde{n},n}^{s}}{t_{\tilde{n},n}^{s}}\left(\frac{1}{2\pi t^{k}_{\tilde{n},n}}-j\frac{1}{\lambda}\right)e^{j 2 \pi t_{\tilde{n},n}^{s}/\lambda}, s \in \mathcal{S},\label{deviationReceiver}
	\end{align}
	where $ A_{r} $ denotes the area of each meta-atom in the CSIM, $ \zeta_{\tilde{n},n}^{s} $ expresses the angle between the propagation direction and the normal direction of the $ (s-1) $th CSIM metasurface layer, and $ t_{\tilde{n},n}^{s} $ is the corresponding transmission distance. In this case, the effect of the CSIM is described by
	\begin{align}
		\bZ=\bLambda^{S}\bU^{S}\cdots\bLambda^{2}\bU^{2}\bLambda^{1}\textcolor{black}{\bU^{1}}	\in \mathbb{C}^{N \times N},\label{ReceiverSIM}
	\end{align}
	where $ \bU^{s}\in \mathbb{C}^{N \times N}, s \in \mathcal{S} $
	is the transmission coefficient matrix between the $ s $th CSIM metasurface layer to the $ (s-1) $st CSIM metasurface layer.
	\textcolor{black}{\begin{remark}
			The  propagation coefficients (obtained by the inter-layer propagation formula) result in an attenuation of the signal. The more layers there are the greater the signal attenuation. Reducing the layer spacing appropriately could help reduce the signal attenuation. For example, with a large layer spacing that causes attenuation $0.5$, doubling the layers would result in \textcolor{black}{$0.5\times 0.5 = 0.25$}. 	However, if we could reduce the layer spacing to get about a $0.9$ attenuation, doubling the layers would result in \textcolor{black}{$0.9 \times 0.9 = 0.81$}. \textcolor{black}{Regarding the signal attenuation,  the cascaded loss associated with each propagation path becomes more severe as the number of layers increases. Nevertheless, each metasurface layer has a large aperture that helps to mitigate this loss. As a result, the overall SIM response is not expected to experience significant signal attenuation. In fact, recent research shows that increasing the number of layers may even lead to a gain increase \cite{Wang2024}.} Further study of the signal attenuation could be the topic of future work.
	\end{remark}}
	
	Based on  narrowband quasi-static block fading channel modeling, we assume that  the duration of each block is $\tau_{\mathrm{c}}$ channel uses. Also, by adopting the time-division-duplex (TDD) protocol as the  recommended protocol for next-generation systems such as mMIMO systems, $\tau$ channel uses are allocated for the uplink training phase and $\tau_{\mathrm{c}}-\tau$ channel uses for the uplink data transmission phase.

	Now, we denote   $ \bG=[\bg_{1}\ldots,\bg_{N} ] \in \mathbb{C}^{M \times N}$  the channel matrix between the BSIM and the CSIM, where  $ \bg_{i} \in \mathbb{C}^{M \times 1}$ for $ i\in \mathcal{N} $. Moreover,		$ \bq_{k} \in \mathbb{C}^{N \times 1}$ expresses  the channel between the CSIM and user $ k $. Also, we assume the presence of a direct link between the BSIM  and user  $ k $. We denote the corresponding link  as $ \bd_{k} $. In this work, we assume the presence of realistic correlated Rayleigh fading, which occurs in practice \cite{Bjoernson2020}.\textcolor{black}{ The  analysis with correlated Rician fading that  includes a line-of-sight (LoS) component could be the topic of future work.} Specifically, we have
	\begin{align}
		\bG&=\sqrt{\tilde{ \beta}_{g}}\bR_{\mathrm{BSIM}}^{1/2}\bD\bR_{\mathrm{CSIM}}^{1/2},\label{eq2}\\
		\bq_{k}&=\sqrt{\tilde{ \beta}_{k}}\bR_{\mathrm{CSIM}}^{1/2}\bc_{k},\\	
		\bd_{k}&=\sqrt{\bar{ \beta}_{k}}\bR_{\mathrm{BSIM}}^{1/2}\bar{\bc}_{k},
	\end{align}
	where $ \bR_{\mathrm{BSIM}} \in \mathbb{C}^{M \times M} $ and $ \bR_{\mathrm{CSIM}} \in \mathbb{C}^{N \times N} $ describe the  correlation matrices at the BSIM and the CSIM respectively. Note that these matrices are deterministic Hermitian-symmetric positive semi-definite and are assumed to be known by the network. The   spatial correlation matrices at the BSIM and CSIM are given by \cite{Bjoernson2020}
	\begin{align}
		[\bR_{\mathrm{BSIM}}]_{m,\tilde{m}}&=\mathrm{sinc}(2 r_{m,\tilde{m}}/\lambda), \tilde{m}\in \mathcal{M}, m\in \mathcal{M},\label{corB}\\
		[\bR_{\mathrm{CSIM}}]_{\tilde{n},n}&=\mathrm{sinc}(2 t_{\tilde{n},n}/\lambda), n\in \mathcal{N}, \tilde{n}\in \mathcal{N},\label{corC}
	\end{align}
	where $ r_{m,\tilde{m}} $ and $ t_{\tilde{n},n} $ are the corresponding meta-atom spacings. The path-losses $\tilde{ \beta}_{g} $, $ \bar{ \beta}_{k} $, and $\tilde{ \beta}_{k} $ correspond to BSIM-CSIM, CSIM-user $ k $, and BSIM-user $ k $ links,  while  $ \bc_{k} \sim \mathcal{CN}\left(\b0,\Id_{N}\right) $,  $ \bar{\bc}_{k} \sim \mathcal{CN}\left(\b0,\Id_{M}\right)$, and $ \mathrm{vec}(\bD)\sim \mathcal{CN}\left(\b0,\Id_{MN}\right) $  represent the small-scale fading parts. Contrary to most previous works, which assume that one of the links of the cascaded channels is deterministic, i.e., it represents a LoS component \cite{Nadeem2020,Papazafeiropoulos2021}, we consider a more general analysis, since we consider  small-scale fading among all links.

	By assuming that the phase shift matrices $ \bPhi^{l}, l \in \mathcal{L}$  and $ \bLambda^{s}, s\in \mathcal{S}$ are given,  the aggregated channel vector from the last BSIM layer to  user $ k $,  includes the cascaded and direct channels, i.e., $ \bh_{k}= \bG \bZ \bq_{k}+\bd_{k} \in \mathbb{C}^{M \times 1} $, which  has a covariance matrix  given by
	$ \bR_{k}=\EE\{\bh_{k}\bh_{k}^{\H}\} $. Specifically, we have
	\begin{align}
		&	\bR_{k}=	\EE\{\bG \bZ \EE\{\bq_{k}\bq_{k}^{\H}\}\bZ^{\H}\bG^{\H}\}+\bar{ \beta}_{k}\bR_{\mathrm{BSIM}}\label{cov1}\\
		&=	\textcolor{black}{\tilde{\beta}_{k}}\EE\{\bG \bZ \bR_{\mathrm{CSIM}}\bU^{1^{\H}}\bZ^{\H}\bG^{\H}\}+\bar{ \beta}_{k}\bR_{\mathrm{BSIM}}\label{cov2}\\
		&\textcolor{black}{=\hat{\beta}_{k}\mathrm{tr}(\bR_{\mathrm{CSIM}} \bZ \bR_{\mathrm{CSIM}}\bZ^{\H}\}\bR_{\mathrm{BSIM}}+\bar{ \beta}_{k}\bR_{\mathrm{BSIM}}}\label{cov30}
	\end{align}
	where, in \eqref{cov1}, we have applied the independence between $ \bG $ and $ \bq_{k} $,  denoted  $\hat{\beta}_{k}= \tilde{ \beta}_{g}\tilde{ \beta}_{k} $, and used $ \EE\{	\bd_{k}	\bd_{k}^{\H}\} =\bar{ \beta}_{k}\bR_{\mathrm{BSIM}} $.  Next, we have used that $ \EE\{	\bq_{k}	\bq_{k}^{\H}\} =\tilde{ \beta}_{k} \bR_{\mathrm{CSIM}}$. In \eqref{cov30},  we have used the property $ \EE\{\bV \bU\bV^{\H}\} =\tr (\bU) \Id_{M}$, where $\bU  $ is a deterministic square matrix, and $ \bV $ is any matrix with independent and identically distributed (i.i.d.) entries of zero mean and unit variance. 

	\section{Channel Estimation and Uplink Data Transmission}\label{ChannelEstimation}
	\subsection{Channel Estimation}
	Given that perfect CSI is not attainable in practice, we employ the TDD protocol that considers an uplink training phase. In this phase, pilot symbols are transmitted to estimate the channel. However, a SIM has the peculiar characteristic that it cannot process any signals, i.e., it cannot receive pilots and it cannot transmit them because it is implemented without any RF chains.
	
	Herein, we propose a channel estimation method, where we estimate the aggregated channel as in \cite{Nadeem2020, Papazafeiropoulos2021,Deshpande2022} instead of the individual channels as in \cite{Yang2020b,Wu2021,Zheng2022}. The latter approach comes with extra power  and hardware cost. Moreover, the large dimension of the BSIM-CSIM assisted channel would require  a  pilot overhead that would be prohibitively high if we estimated the  individual channels. Also, we manage to express the estimated channel in closed form. Hence, the proposed method follows.
	
	During the  training phase, which lasts $ \tau $ channel uses, it is assumed that  all users  send orthogonal pilot sequences. In particular, $\bx_{k}=[x_{k,1}, \ldots, x_{k,\tau}]^{\H}\in \mathbb{C}^{\tau\times 1} $ with $ \bx_{k}^{\H}\bx_{l}=0~\forall k\ne l$ and $ \bx_{k}^{\H}\bx_{k}= \tau \rho$ joules denotes the pilot sequence of user $ k  $, where  all users use the same normalized signal-to-noise ratio (SNR) $ \rho $ for transmitting each pilot symbol during 	the training phase.
	
	Overall, the BS receives 
	\begin{align}
		\bY^{\tr}=\sum_{i=1}^{K} \bW^{1^{\H}}\bP^{H}\bh_{i}\bx_{i}^{\H} +
		\bZ^{\tr},\label{train1}
	\end{align}
	with  $ \bZ^{\tr} \in \mathbb{C}^{M \times \tau} $ being the received AWGN matrix having independent columns with each one distributed as $ \mathcal{CN}\left(\b0,\Id_{M}\right)$. 
	
	By multiplying \eqref{train1} with $ \bx_{k} $, we obtain
	\begin{align}
		\br_{k}=\bc_{k}+\frac{\bz_{k}}{ \tau \rho},\label{train2}
	\end{align}
	where $ \bc_{k}=\bW^{1^{\H}}\bP^{H}\bh_{k} $ and  $ \bz_{k}=\bZ^{\tr} \bx_{k}$.

	Based on linear MMSE estimation,   the overall perfect channel can be written as
	\begin{align}
		\bc_{k}=\hat{\bc}_{k}+\tilde{\bc}_{k}\label{current}, 
	\end{align}
	where $\hat{\bc}_{k}$ and  $\tilde{\bc}_{k}$ are the channel estimate and  estimation channel error vectors, respectively. The following lemma provides the linear minimum mean-squared error (MMSE)  estimated channel. Note that $\hat{\bc}_{k}$ and $\tilde{\bc}_{k}$ are uncorrelated and each of them has zero mean, but they are not independent~\cite{Bjoernson2017}.

	\begin{lemma}\label{LemmaDirectChannel}
		The linear MMSE estimated channel $ \bc_{k} $ of  user $ k $ at the BS  is obtained as
		\begin{align}
			\hat{\bc}_{k}=\hat{\bR}_{k}\bQ_{k} \br_{k},\label{estim1}
		\end{align}
		where $ \hat{\bR}_{k}=\bW^{1^{\H}}\bP^{H}\bR_{k}\bW^{1}\bP $, $ \bQ_{k}\!=\! \left(\!\hat{\bR}_{k}\!+\!\frac{1}{ \tau \rho }\Id_{M}\!\right)^{\!-1}$, and $ \br_{k}$ is the noisy channel given by \eqref{train2}. 
	\end{lemma}
	\begin{proof}
		Please see Appendix~\ref{Lemma1}.	
	\end{proof}
	
	The covariances of \textcolor{black}{$\hat{\bc}_{k}$ and  $\tilde{\bc}_{k}$} are easily obtained as
	\begin{align}
		\bPsi_{k}\!&=\!\hat{\bR}_{k}\bQ_{k}\hat{\bR}_{k},\label{Psiexpress}\\
		\tilde{\bPsi}_{k}&=\hat{\bR}_{k}-\bPsi_{k}.\label{Psiexpress1}
	\end{align}
	
	As can be seen, in this method, the  length of the pilot sequence depends only on the number of users $ K $, and is independent of the number of surfaces and elements at the BSIM and CSIM, which are generally large. Thus, we achieve substantial overhead reduction compared to estimating the individual channel, where the complexity would increase with increasing SIMs sizes. Also, we achieve to obtain the estimated overall channel  in closed form instead of  works such as \cite{Nadeem2023} that cannot result in analytical expressions.
	
	Another benefit of the proposed two-timescale  method is that the channel is required  to be estimated  when the large-scale statistics (statistical CSI) varies, which is at every several coherence intervals. On the contrary, instantaneous CSI-based approaches would require channel estimation at every coherence interval that would result in high power consumption and computational complexity.

	\subsection{Uplink Data Transmission}
	In this subsection, we provide the derivation of the uplink achievable sum SE of a SIM-assisted mMIMO system.
	
	During the uplink data transmission phase, the BS receives
	\begin{align}
		\by=\sqrt{\rho}\sum_{i=1}^{K}\bW^{1^{\H}}\bP^{H}\bh_{i}s_{i} +\bn,\label{ULTrans}
	\end{align}
	where $\bn\sim \mathcal{CN}\left(\b0,\Id_{M}\right) $ is the AWGN at the BS, $ s_{i}\sim \mathcal{CN}(0,1) $ is the transmit signal vector, $ \rho $ is the uplink SNR, and $ \bh_{i} $ is the aggregated channel from the last BSIM layer to  user $ i $.
	
	After application of  a receive combining vector $ \bv_{k} $ as $ \hat{s}_{k}= \bv_{k}^{\H}\by_{}$, the BS obtains $ s_{k} $ from user $ k $ as
	\begin{align}
		\hat{s}_{k}&=\sqrt{\rho}\EE\left\{\bv_{k}^{\H}\bW^{1^{\H}}\bP^{H}{\bh}_{k}\right\} s_{k} \nn\\
		&+ \sqrt{\rho}\left(\bv_{k}^{\H}\bW^{1^{\H}}\bP^{H}{\bh}_{k}-\EE\left\{\bv_{k}^{\H}\bW^{1^{\H}}\bP^{H}{\bh}_{k}\right\}\right) s_{k}\nn\\
		&+\sum_{i\ne k}^{K}\sqrt{\rho}\bv_{k}^{\H}\bW^{1^{\H}}\bP^{H}{\bh}_{i}s_{i}+ \bv_{k}^{\H}\bn,\label{received1}
	\end{align}
	where we have also accounted for channel hardening, which is applicable in mMIMO systems with $ M>8 $ \cite{Bjoernson2017}. In \eqref{received1}, the first term is the desired signal,  the second term is the beamforming gain uncertainty, the third term is multiuser interference, and the last term is the post-processed AWGN noise.

	Next, we derive the use-and-then-forget (UaTF) bound on the uplink average SE, which is a lower bound \cite{Bjoernson2017}. For this reason, we assume the worst-case uncorrelated additive noise scenario for the inter-user interference and the AWGN noise. In this case, the achievable SE can be written as 
	\begin{align}
		\mathrm{SE}_{k}	=\frac{\tau_{\mathrm{c}}-\tau}{\tau_{\mathrm{c}}}\log_{2}\left ( 1+\gamma_{k}\right)\!,\label{LowerBound}
	\end{align}
	where the fraction represents the percentage of samples transmitted during the uplink data transmission, and  $ \gamma_{k}=\frac{S_{k}}{I_{k}}$ denotes  the uplink signal-to-interference-plus-noise ratio (SINR). Here, $ S_{k}$ and  $I_{k} $
	describe the desired signal power and the interference plus noise power. In terms of $ \gamma_{k} $, we have
	\begin{align}
		S_{k}&=\rho|\EE\left\{\bv_{k}^{\H}\bc_{k}\right\}|^{2},\label{sig11}\\
		I_{k}&=\EE\left\{|\bv_{k}^{\H}\bc_{k}-\EE\left\{\bv_{k}^{\H}\bc_{k}\right\}|^{2}\right\}+\sum_{i\ne k}^{K}\EE\left\{|\bv_{k}^{\H}\bc_{i}|^{2}\right\}
		\nn\\&+\!\EE\left\{|\bv_{k}^{\H}\bn|^{2}\right\}\!\label{int1}.
	\end{align}

	In this work, we assume MRC decoding to provide a closed-form SE, i.e., $ \bv_{k}=\hat{\bc}_{k}$.\footnote{The  application of the MMSE decoder, which is another common linear receiver used for mMIMO systems due to its better performance, is the topic of ongoing work. The reason is that the corresponding analysis requires the use of deterministic equivalent tools \cite{Hoydis2013,Papazafeiropoulos2015a,Papazafeiropoulos2016}, which result in complicated expressions that would result in an intractable optimization regarding the phase shifts matrices including in  BSIM and CSIM. }

	\begin{theorem}\label{theorem:ULSINR}
		The uplink  SINR of user $k$ for given SIMs in a BSIM and CSIM-assisted mMIMO system is tightly approximated as
		
		\begin{equation}
			\gamma_{k} \ {\color{black}\approx}\ 	\frac{S_{k}}{	\color{black}\tilde{I}_{k}},\label{gammaSINR}
		\end{equation}
		where
		\begin{align}
			{S}_{k}&=\tr^{2}\left(\bPsi_{k}\right)\!,\label{Num1}\\
			\tilde{I}_{k}&=\sum_{i =1}^{K}
			\tr\!\left( \hat{\bR}_{i}\bPsi_{k}\right)-\tr\left( \bPsi_{k}^{2}\right)+\frac{1}{\rho}\tr(\bPsi_{k}).\label{Den1}
		\end{align}
	\end{theorem} 
	\begin{proof}
		Please see Appendix~\ref{theorem1}.	
	\end{proof}
	
	Theorem \ref{theorem:ULSINR} provides a tight approximation as the  numerical results in Section \ref{Numerical} reveal. The tightness is also expected because it comes as a result of  the exploitation of  channel hardening  appearing in mMIMO systems~\cite{Bjoernson2017}.

	\section{Problem Formulation  and  Optimization}\label{PSConfig}
	It is crucial to optimize the sum SE in terms of the phase shift matrices of the SIMs included in the proposed architecture.
	
	\subsection{Problem Formulation}
	By assuming infinite-resolution phase shifters, the formulation of the optimization problem follows as 
	\begin{subequations}\label{eq:subeqns}
		\begin{align}
			(\mathcal{P})~~&\max_{\bphi_{l},\blambda_{s}} 	\;	f(\bphi_{l},\blambda_{s})=\sum_{k=1}^{K}\log_2\left(1+\frac{S_k}{\tilde{I}_k}\right)\label{Maximization1} \\
			&~	\mathrm{s.t}~~~	\bP=\bPhi^{L}\bW^{L}\cdots\bPhi^{2}\bW^{2}\bPhi^{1}\bW^{1},
			\label{Maximization3} \\
			&\;\quad\;\;\;\;\;\!\!~\!		\bZ=\bLambda^{K}\bU^{K}\cdots\bLambda^{2}\bU^{2}\bLambda^{1}\bU^{1},
			\label{Maximization4} \\
			&\;\quad\;\;\;\;\;\!\!~\!		\bPhi^{l}=\diag(\phi^{l}_{1}, \dots, \phi^{l}_{M}), l \in \mathcal{L},
			\label{Maximization5} \\
			&\;\quad\;\;\;\;\;\!\!~\!		\bLambda^{s}=\diag(\lambda^{s}_{1}, \dots, \lambda^{s}_{N}), k \in \mathcal{S},
			\label{Maximization6} \\
			&\;\quad\;\;\;\;\;\!\!~\!		|	\phi^{l}_{m}|=1, m \in \mathcal{M}, l \in \mathcal{L},	\label{Maximization7} \\
			&	\;\quad\;\;\;\;\;\!\!~\!		|\lambda^{s}_{n}|=1, n \in \mathcal{N}, s \in \mathcal{S}	\label{Maximization8}.
		\end{align}
	\end{subequations}
	
	In problem $ 	(\mathcal{P}) $, we have denoted $ f(\bphi_{l},\blambda_{s}) $ the objective function that includes the approximate of the SINR provided by Theorem \eqref{theorem:ULSINR}.

	\subsection{Simultaneous Optimization of Both SIMs}	
	The non-convexity  of problem $(\mathcal{P})  $	and the coupling between  the optimization variables of the two SIMs, which also have constant modulus constraints lead us to seek a locally optimal solution.  Note that many previous  optimization methods on RIS-aided systems have applied alternating optimization (AO) to optimize the relevant variables in each case \cite{Wu2019,Zhang2020a}. Although AO-based methods can be implemented easily, their convergence may require many iterations, which increase with increasing the size of RIS \cite{Perovic2021}. Since SIM-assisted systems are expected to include a large number of elements, the AO method is not suggested.
	
	Given that the projection operators can be obtained in closed form, we can apply the  projected  gradient ascent method (PGAM) \cite[Ch. 2]{Bertsekas1999} for the simultaneous optimization of $ \bphi_{l} $ and $ \blambda_{s} $ instead of the AO method. Numerical results in Section \ref{Numerical} will demonstrate the superiority of the proposed approach.

	\subsection{Proposed Algorithm-PGAM}	
	According to the proposed PGAM, starting from $ (\bphi_{l}^{0},\blambda_{l}^{0}) $, we shift towards the gradient of the objective function, i.e., $ \nabla f(\bphi_{l},\blambda_{s}) $, where  the rate of change of $ f(\bphi_{l},\blambda_{s}) $ becomes maximum. The newly computed points $ \bphi_{l} $ and  $\blambda_{s} $ are projected onto  $ \Phi_{l} $ and $	\Lambda_{s} $, respectively to keep the updated points inside the feasible set. 
	
	The step size of this shift is determined by a parameter $ \mu>0 $. For the sake of description, we define the following sets.
	\begin{align}
		\Phi_{l}&=\{\bphi_{l}\in \mathbb{C}^{M \times 1}: |\phi^{l}_{i}|=1, i=1,\ldots, M\},\\
		\Lambda_{s}&=\{\blambda_{s}\in \mathbb{C}^{N \times 1}: |\lambda^{s}_{i}|=1, i=1,\ldots, N\}.
	\end{align}

	The proposed PGAM solving \eqref{eq:subeqns} is outlined in Algorithm \ref{Algoa1}. As can be seen, the algorithm includes the following two iterations.
	\begin{align}
		\bphi_{l}^{i+1}&=P_{\Phi_{l}}(\bphi_{l}^{i}+\mu_{i}\nabla_{\bphi_{l}}f(\bphi_{l}^{i},\blambda_{s}^{i}))\label{p1}\\
		\blambda_{s}^{i+1}&=P_{\Lambda_{s}}(\blambda_{s}^{i}+\mu_{i}\nabla_{\blambda_{s}}f(\bphi_{l}^{i},\blambda_{s}^{i})), \label{p2}
	\end{align}
	where $ P_{\Phi_{l}}(\cdot) $ and $ P_{\Lambda_{s}}(\cdot) $ are the projections onto $ \Phi_{l} $ and $ \Lambda_{s} $, respectively.  Also, to  find the step size, we apply the Armijo-Goldstein backtracking line search. For this reason, we  define the quadratic approximation of $ f(\bphi_{l},\blambda_{s}) $ as
	\begin{align}
		&	Q_{\mu}(\bphi_{l}, \blambda_{s};\bx,\by)=f(\bphi_{l},\blambda_{s})+\langle	\nabla_{\bphi_{l}}f(\bphi_{l},\blambda_{s}),\bx-\bphi_{l}\rangle\nn\\
		&-\frac{1}{\mu}\|\bx-\bphi_{l}\|^{2}_{2}+\langle\nabla_{\blambda_{s}}f(\bphi_{l},\blambda_{s}),\by-\blambda_{s}\rangle\nn\\
		&-\frac{1}{\mu}\|\by-\blambda_{s}\|^{2}_{2}.
	\end{align}
	
	The step size in \eqref{p1} and \eqref{p2} can be found as $ \mu_{i} = L_{i}\kappa^{m_{i}} $ with  $ L_i>0 $, and $ \kappa \in (0,1) $, where the parameter $ m_{n} $ is the
	smallest nonnegative integer satisfying
	\begin{align}
		f(\bphi_{l}^{i+1},\blambda_{s}^{i+1})\geq	Q_{L_{i}\kappa^{m_{i}}}(\bphi_{l}^{i}, \blambda_{s}^{i};\bphi_{l}^{i+1},{\color{black}\blambda_{s}^{i+1}}),
	\end{align}
	which can be performed by an iterative procedure.

	\begin{algorithm}[th]
		\caption{Simultaneous PGAM for both SIMS\label{Algoa1}}
		\begin{algorithmic}[1]
			\STATE Input: $\bphi_{l}^{0},\blambda_{s}^{0},\mu_{1}>0$, $\kappa\in(0,1)$
			\STATE $i\gets1$
			\REPEAT
			\REPEAT \label{ls:start}
			\STATE $\bphi_{l}^{i+1}=P_{\Phi_{l}}(\bphi_{l}^{i}+\mu_{ni}\nabla_{\bphi_{l}}f(\bphi_{l}^{i},\blambda_{s}^{i}))$
			\STATE $\blambda_{s}^{i+1}=P_{\Lambda_{s}}(\blambda_{s}^{i}+\mu_{i}\nabla_{\blambda_{s}}f(\bphi_{l}^{i},\blambda_{s}^{i}))$
			\IF{ $f(\bphi_{l}^{i+1},\blambda_{s}^{i+1})\leq Q_{\mu_{i}}(\bphi_{l}^{i},\blambda_{s}^{i};\bphi_{l}^{i+1},{\color{black}\blambda_{s}^{i+1}})$}
			\STATE $\mu_{i}=\mu_{i}\kappa$
			\ENDIF
			\UNTIL{ $f(\bphi_{l}^{i+1},\blambda_{s}^{i+1})>Q_{\mu_{i}}(\bphi_{l}^{i},\blambda_{s}^{i};\bphi_{l}^{i+1},\blambda_{s}^{i+1}})$\label{ls:end}
			\STATE $\mu_{i+1}\leftarrow \mu_{i}$
			\STATE $i\leftarrow i+1$
			\UNTIL{ convergence}
			\STATE Output: $\bphi_{l}^{i+1},\blambda_{s}^{i+1}$
		\end{algorithmic}
	\end{algorithm}

	\subsection{Complex-valued Gradients}
	Below, we present the  gradients   $\nabla_{\bphi_{l}}f(\bphi_{l},\blambda_{s}) $ and $ \nabla_{\blambda_{s}}f(\bphi_{l},\blambda_{s}) $.	
	
	\begin{proposition}\label{propositionGradient}
		The gradients of $f(\bphi_{l},\blambda_{s}) $ regarding  $\bphi_{l}^{*}$ and $\blambda_{s}^{*}$  are provided 
		closed-forms by
		
		\begin{align}
			\nabla_{\bphi_{l}}f(\bphi_{l},\blambda_{s})&=\frac{\tau_{c}-\tau}{\tau_{c}\log_{2}(e)}\sum_{k=1}^{K}\frac{\tilde{I}_{k}\nabla_{\bphi_{l}}S_{k}-S_{k}\nabla_{\bphi_{l}}\tilde{I}_{k}}{(1+\gamma_{k})\tilde{I}_{k}},\label{grad11}\\
			\nabla_{\blambda_{s}}f(\bphi_{l},\blambda_{s})&=\frac{\tau_{c}-\tau}{\tau_{c}\log_{2}(e)}\sum_{k=1}^{K}\frac{\tilde{I}_{k}\nabla_{\blambda_{s}}S_{k}-S_{k}\nabla_{\blambda_{s}}\tilde{I}_{k}}{(1+\gamma_{k})\tilde{I}_{k}},\label{grad12}
		\end{align}
		where
		\begin{align}
			\nabla_{\bphi_{l}}S_{k} 
			&=2\tr(\boldsymbol{\Psi}_{k})\diag(\bA_{l}^{\H}\bR_{k}\bW^{1}\bP\bB_k\bW^{1^{\H}}\bC_{l}^{\H})\label{differentialPhi71},\\
			\nabla_{\bphi_{l}}\tilde{I}_{k}
			&=\diag\big(\bA_{l}^{\H}\big(\bR_{k}\bW^{1}\bP\big({\bPsi}+\bar{\bPsi}_{k}\bB_k\big)	\nn\\
			&+\sum\nolimits _{i=1}^{K}\bR_{i}\bW^{1}\bP\bB_i\big)\bW^{1^{\H}}\bC_{l}^{\H}\big)\label{gradI10}
		\end{align}
		with $\boldsymbol{\Psi}=\sum\nolimits _{i=1}^{K}\boldsymbol{\Psi}_{i}$, $\bar{\boldsymbol{\Psi}}_{k}=\frac{1}{\rho}\Id-2\boldsymbol{\Psi}_{k}$, $\bB_{k}= \mathbf{Q}_{k}\hat{\mathbf{R}}_{k}-\mathbf{Q}_{k}\hat{\mathbf{R}}^{2}_{k}\mathbf{Q}_{k}+\hat{\mathbf{R}}_{k}\mathbf{Q}_{k} $, $ 	\bA_{l}=\bPhi^{L}\bW^{L}\cdots\bPhi^{l+1}\bW^{l+1} $,  $\bC_{l}= \bW^{l}\bPhi^{l-1}\bW^{l-1}\cdots \bPhi^{1} $, 
		and
		\begin{align}
			\nabla_{\blambda_{s}}S_{k}
			&=2\hat{\beta}_{k}\tr(\boldsymbol{\Psi}_{k})\tr(\bW^{1^{\H}}\bP^{H}\bR_{\mathrm{BSIM}}\bW^{1}\bP\bB_{k})	\nn\\
			&\times\diag(\bF_{s}{\bR}_{\mathrm{CSIM}}\bZ^{\H}\bR_{\mathrm{CSIM}}\bD_{s})^{\H},\\
			\nabla_{\blambda_{s}}\tilde{I}_{k}&=
			2\tr\big(\bW^{1^{\H}}\bP^{H}\bR_{\mathrm{BSIM}}\bW^{1}\bP	\nn\\
			&\times(\sum\nolimits _{i=1}^{K}\hat{\beta}_{i}\hat{\mathbf{R}}_{k}\bB_{i}+\hat{\beta}_{k}(\bar{\boldsymbol{\Psi}}_{k}\bB_{k}+\boldsymbol{\Psi})) \big)&\nn\\
			& \times\diag(\bF_{s}{\bR}_{\mathrm{CSIM}}\bZ^{\H}\bR_{\mathrm{CSIM}}\bD_{s})^{\H}
		\end{align}
		with $ \bD_{s}= \bU^{s-1}\bLambda^{s-1}\cdots\bU^{2}\bLambda^{1}$ and $ \bF_{s}= \bLambda^{S}\bU^{S}\cdots\bLambda^{s+1}\bU^{s+1} $.
	\end{proposition}
	\begin{proof}
		Please see Appendix~\ref{proposition1}.	
	\end{proof}

	\subsection{Projection Operations of PGAM}
	The constraint $ |	\phi^{l}_{m}|=1 $  suggests that $ \phi^{l}_{m} $ should be located on the unit circle in the complex plane. Specifically, regarding the  vector $ \bar{\bu}_{l} $ of $ P_{\Phi_{l}}(\bu_{l}) $, we have
	\begin{align}
		\bar{u}_{l,m}=\left\{
		\begin{array}{ll}
			\frac{u_{l,m}}{|u_{l,m}|} & u_{l,m}\ne 0 \\
			e^{j \phi^{l}_{m}}, \phi^{l}_{m} \in [0, 2 \pi] &u_{l,m}=0 \\
		\end{array}, m=1, \ldots,M,
		\right.
	\end{align}
	where   $ \bu_{l} \in \mathbb{C}^{M \times 1} $ is a given point. The projection concerning  $ 	|\lambda_{s}^{n}|=1 $ follows the same lines. 
	
	\subsection{Convergence Analysis}
	In this section, we start with the presentation of the convergence of Algorithm \ref{Algoa1}. Specifically, given that the gradients $ \nabla_{\bphi_{l}}f(\bphi_{l},\blambda_{s}) $ and $ \nabla_{\blambda_{s}}f(\bphi_{l},\blambda_{s}) $ are Lipschitz continuous\footnote{A function $ f(\bx) $ is $ \mathcal{L} $-Lipschitz continuous, or else $ \mathcal{L}  $-smooth over a set $ \mathcal{X} $, if for all $ \bx,\by \in \mathcal{X} $, we have
		\begin{align}
			\|\nabla f(\by)-\nabla f(\bx)\|\le \mathcal{L} \|\by-\bx\|.
	\end{align}} because they consist of basic functions  as shown previously.
	In this case, we have \cite[Chapter 2]{Bertsekas1999}
	\begin{align}
		&	f(\bx,\by)\geq f(\bphi_{l},\blambda_{s})
		+\langle	\nabla_{\bphi_{l}}f(\bphi_{l},\blambda_{s}),\bx-\bphi_{l}\rangle	\nn\\
		&-\frac{1}{L_{\bphi_{l} }}\|\bx-\bphi_{l}\|^{2}_{2}
		+\langle\nabla_{\blambda_{s}}f(\bphi_{l},\blambda_{s}),\by-\blambda_{s}\rangle-\frac{1}{L_{\blambda_{s}}}\|\by-\blambda_{s}\|^{2}_{2}\nn\\
		&\geq f(\bphi_{l},\blambda_{s})
		+\langle	\nabla_{\bphi_{l}}f(\bphi_{l},\blambda_{s}),\bx-\bphi_{l}\rangle	\nn\\
		&-\frac{1}{L_{\max }}\|\bx-\bphi_{l}\|^{2}_{2}\nn
		+\langle\nabla_{\blambda_{s}}f(\bphi_{l},\blambda_{s}),\by-\blambda_{s}\rangle	\nn\\
		&-\frac{1}{L_{\max}}\|\by-\blambda_{s}\|^{2}_{2},		
	\end{align}
	where $L_{\max}=\max(L_{\bphi_{l} },L_{\blambda_{s}})$. Steps $ 4-10 $ of the algorithm demonstrate its termination in finite iterations because the condition in Step $ 10 $ must be fulfiled when $\mu_i <L_{\max}$. The line search results in $f(\bphi_{l}^{i+1},\blambda_{s}^{i+1})\geq f(\bphi_{l}^{i},\blambda_{s}^{i})$, which is an increasing sequence of objectives. Given that the sets $ \Phi_{l} $ and $ \Lambda_{s}$ are compact, $ f(\bphi_{l}^{i},\blambda_{s}^{i}) $ converges. It is worthwhile to mention that this is not necessarily an optimal solution because the problem is nonconvex, but a locally optimal solution.
	
	\subsection{Complexity Analysis}
	Herein, we elaborate on the complexity (number of complex multiplications)  per iteration of the proposed PGAM by using the big-$ \mathcal{O} $ notation, which is quite meaningful for large $ M $ and $ N $ used by the proposed architecture. In particular, for the computation of $ 
	\bP $ and $ 
	\bZ $, we require $ (L-1) M^{3}+M^{2}M_{\mathrm{BS}}$ and  $ (S-1) N^{3}+N^{2}$, multiplications, respectively.  The computation of $ 	\bR_{k} $, which includes a trace over $ \bR_{\mathrm{CSIM}}\bZ {\bR}_{\mathrm{CSIM}}\bZ^{\H} $, requires $ M^{2}+N^{2} $ multiplications. Note that the former $M^{2}  $ multiplications are required for the multiplication between the trace and $ \bR_{\mathrm{BSIM}} $. The complexity of $ \hat{\bR}_{k} $ is $ \mathcal{O}(M^{3}+M^{2}) $. From \cite{Papazafeiropoulos2023}, it can be observed that the complexity of $ 	\bPsi_{k} $, which includes the inverse matrix $ 	\bQ_{k} $, can be limited to $ \mathcal{O}(M^{2}) $. Hence, the complexity to compute the objective function is $ \mathcal{O}(K((L-1) M^{3}+(S-1) N^{3}+M^{2}M_{\mathrm{BS}}+N^{2})) $.
	
	Next, we provide the complexity for the computation of the gradients $ 	\nabla_{\bphi_{l}}f(\bphi_{l},\blambda_{s}) $ and $ 	\nabla_{\blambda_{s}}f(\bphi_{l},\blambda_{s}) $. Each of the   gradients $ 	\nabla_{\bphi_{l}}S_{k}  $ and $ 	\nabla_{\bphi_{l}}\tilde{I}_{k} $ require $ M^{2}+N^{2} $  multiplications. A similar complexity is exhibited by $ \nabla_{\blambda_{s}}S_{k} $ and $\nabla_{\bphi_{l}}\tilde{I}_{k}  $. Thus, the total complexity for obtaining the gradients is $  \mathcal{O}( K(M^{2}+N^{2}))$.
	
	\section{Numerical Results}\label{Numerical}
	In this section, we present the evaluation of the achievable sum SE under the proposed  algorithm optimizing simultaneously both BSIM and CSIM in terms of analytical results and Monte Carlo simulations.
	
	\subsection{Simulation Setup}
	We consider the uplink of a  BSIM and CSIM-assisted mMIMO system, where the BS, equipped with a uniform linear array (ULA), consists of $ M_{\mathrm{BS}}=32 $ antennas that are positioned along the $ x$-axis. The BSIM  is integrated with the BS to apply to receive beamforming in the EM wave domain. The BSIM is parallel to the $ x-y $ plane, and its center aligns with the $ z-$axis at a height $ H_{\mathrm{BS}}=10~\mathrm{m} $. The location of the CSIM is at  $(x_{\mathrm{CSIM}},~ y_{\mathrm{CSIM}})=(50,~ 10)$. Also, for the sake of simplicity, we assume that all surfaces are isomorphic and square, which means that the same number of meta-atoms are located along the $ x-$ and $ y-$axes, denoted by $ M_{x} $ and $ M_{y} $, respectively. Note that the spacing  between adjacent antennas/meta-atoms for the BS and metasurfaces is assumed to be $ \lambda/2 $. The size of each meta-atom for both SIMs is $ \lambda/2 \times \lambda /2 $. Each BS antenna has a gain of $ 5 \mathrm{dBi} $ while each user has a single antenna with a gain of $ 0 \mathrm{dBi} $ \textcolor{black}{\cite{Papazafeiropoulos2021,An2023,An2023b}}. The thickness of both SIMs is $ T_{\mathrm{SIM}}=5 \lambda $, while the spacings for the BSIM and CSIM are $ d_{\mathrm{BSIM}}= T_{\mathrm{SIM}}/L$ and $ d_{\mathrm{CSIM}}= T_{\mathrm{SIM}}/S$, respectively. The users  are located on a straight line between $(x_{\mathrm{CSIM}}-\frac{1}{2}d_0,~y_{\mathrm{CSIM}}-\frac{1}{2}d_0)$ and $(x_{\mathrm{CSIM}}+\frac{1}{2}d_0,~y_{\mathrm{CSIM}}-\frac{1}{2}d_0)$ with equal distances between each two adjacent users, and $d_0 = 20$~m.
	
	The  spatial attenuation coefficients $ 	w_{m,\tilde{m}}^{l} $ and $ 	u_{\tilde{n},n}^{s} $ between adjacent metasurface layers in BSIM and CSIM are obtained by \eqref{deviationTransmitter} and \eqref{deviationReceiver}, respectively. The distance between the $ \tilde{m}-$th meta-atom of the $ (l-1)-$st metasurface and the  $ {m}-$th meta-atom of the $ l-$st metasurface is given by $ d_{m,\tilde{m}}^{l}=\sqrt{d_{\mathrm{BSIM}}^{2}+d_{m,\tilde{m}}^{2}} $, where
	\begin{align}
		\!\!	d_{m,\tilde{m}}\!=\!\frac{\lambda}{2}\sqrt{\lfloor | m-\tilde{m}|/M_{x}\rfloor^{2}+[ \mathrm{mod}(|m-\tilde{m}|,M_{x})]^{2}}.
	\end{align}
	Moreover, the transmission distance between the $ m $-th antenna and the $ \tilde{m} $-th meta-atom on the first metasurface layer is provided by \eqref{dist1}. Note that we have $ \cos x_{m,\tilde{m}}^{l}= d_{\mathrm{BSIM}}/ d_{m,\tilde{m}}^{l}, \forall l \in \mathcal{L}$.
	
	\begin{figure*}
		\begin{align}
			{\small 	d_{\tilde{m},m}^{1}\!=\!\sqrt{\!d_{\mathrm{BSIM}}^{2}\!+\!\Big[\!\Big(\!\mathrm{mod}(\tilde{m}\!-\!1, M_{x})\!-\!\frac{M_{x}\!-\!1}{2}\!\Big)\frac{\lambda}{2}\!-\!\Big(m\!-\!\frac{M_{\mathrm{BS}}\!+\!1}{2}\Big)\frac{\lambda}{2}\Big]^{2}\!+\!\Big(\!\lceil \tilde{m}/M_{x} \rceil\!-\!\frac{M_{y}\!+\!1}{2}\Big)^{2}\frac{\lambda_{2}}{4}}}.\label{dist1}
		\end{align}
		\hrulefill
	\end{figure*}
	Similarly, the distance between the $ \tilde{n}-$th meta-atom of the $ (s-1)-$st metasurface and the  $ {n}-$th meta-atom of the $ s-$st metasurface is given by $ d_{n,\tilde{n}}^{s}=\sqrt{d_{\mathrm{CSIM}}^{2}+d_{n,\tilde{n}}^{2}} $, where
	\begin{align}
		d_{n,\tilde{n}}=\frac{\lambda}{2}\sqrt{\lfloor | n-\tilde{n}|/N_{x}\rfloor^{2}+[ \mathrm{mod}(|n-\tilde{n}|,N_{x})]^{2}}.
	\end{align}
	where $ N_{x} $ and $ N_{y} $ are the numbers of meta-atoms along the $ x-$ and $ y-$axes of the CSIM.

	The  path-loss is distance-dependent, and is given by 
	\begin{align}
		\tilde \beta_g = C_{0} (d_g/\hat{d})^{-\alpha_{g}},
	\end{align} 
	where $ C_{0} =(\lambda_{2}/4 \pi \hat{d})$ is  the free space path loss at the reference distance of $ \hat{d}=1~ \mathrm{m}$, $ d_g  $ is the distance between the BSIM and the CSIM.  Regarding $ \tilde{ \beta}_{k} $ and $ \bar{ \beta}_{k} $,  we assume that the corresponding path-loss exponents are equal to $ 2.8 $ and $ 3.5 $, respectively, while  $ \tilde{ d}_{k} $ and $ \bar{ d}_{k} $ are the corresponding distances in the place of $ d_g $ \textcolor{black}{\cite{Papazafeiropoulos2021,An2023,An2023b}.  Notably, in Fig. \ref{fig3}.(a) below, we have used the path-loss exponents $2.29$ and $1.88$ from \cite[Table VI]{Sang2024}, which come from measurements of RIS channels.} The correlation matrices $ \bR_{\mathrm{BSIM}}$ and $\bR_{\mathrm{CSIM}} $ are  calculated according to \eqref{corB} and \eqref{corC}, respectively.  We assume that the  carrier frequency and the system bandwidth are $ 2~\mathrm{GHz} $ and $ 20~\mathrm{MHz} $, respectively. The coherence bandwidth is asummed to be $ B_{\mathrm{c}}  = 200~\mathrm{KHz}$, and the coherence interval is $  \tau_{\mathrm{c}} = 200~\mathrm{symbols} $, given  a coherence time equal to $ T_{\mathrm{c}} = 1~\mathrm{ms} $. Also, we assume
	$ \tau = K $ orthonormal pilot sequences. The uplink SNR for both training and data transmission phases is $ 6~\mathrm{dB}$. Furthermore, we assume  $ K=4 $, $ M=N=100 $, and  $ L=S=4 $.

	\subsection{Convergence of PGAM}
	In Fig. \ref{fig2}.(a), we show the convergence of the proposed algorithm and its superiority for various sets of  meta-atoms per metasurface against the AO method, where the BSIM and CSIM are optimized independently in an alternating way. Termination of the Algorithm \ref{Algoa1} takes place when
	the difference of the objective between the two last iterations
	is less than $10^{-5} $ or the number of iterations is larger than
	$ 100 $. We consider three cases, which are $ M=N=100 $,  $ M=100,N=49 $, and $ M=N=49 $. In all cases, the PGAM converges to its maximum as the number of iterations increases. We observe that an increase in the numbers of meta-atoms requires a higher number of iterations for convergence due to expanded search space since now the corresponding matrices are larger. Also, the increases in the numbers of meta-atoms result in higher complexity of Algorithm \ref{Algoa1}. As $ M, N $ increase both methods approach the optimum SE slower. Moreover, it is shown that the PGAM converges quicker to the optimum SE compared to the AO method, which requires more iterations to reach convergence. Note that the SE for the AO is higher at the beginning than the PGAM because the starting point for the phase shifts leads to higher SE, while, for the PGAM, the initial rate is a result of the initial values of the phase shifts as described by the algorithm.

	The nonconvexity of Problem $ (\mathcal{P}) $  does not allow PGAM to guarantee a stationary solution, which is not necessarily optimal. Hence,  Algorithm \ref{Algoa1} may converge to different points starting from different initial points.  In other words, to mitigate this performance  sensitivity of Algorithm \ref{Algoa1} on the initial points, 
	it is desirable to select the best   solutions after executing the algorithm from different initial points.  For this reason, in Fig. \ref{fig2}.(b), we illustrate the sum SE versus  the iteration count obtained by Algorithm \ref{Algoa1} for $ 5 $ different randomly generated initial points. As can be seen, all $ 5 $ randomly generated initial points result in the same SE. Generally, it is good to follow this observation and execute Algorithm $ 1 $ for $ 5 $ randomly generated initial points to allow a good trade-off between performance and complexity.
	
	\begin{figure}%
		\centering
		\subfigure[\textcolor{black}{Varying sets of numbers meta-atoms per metasurface.}]{	\includegraphics[width=0.9\linewidth]{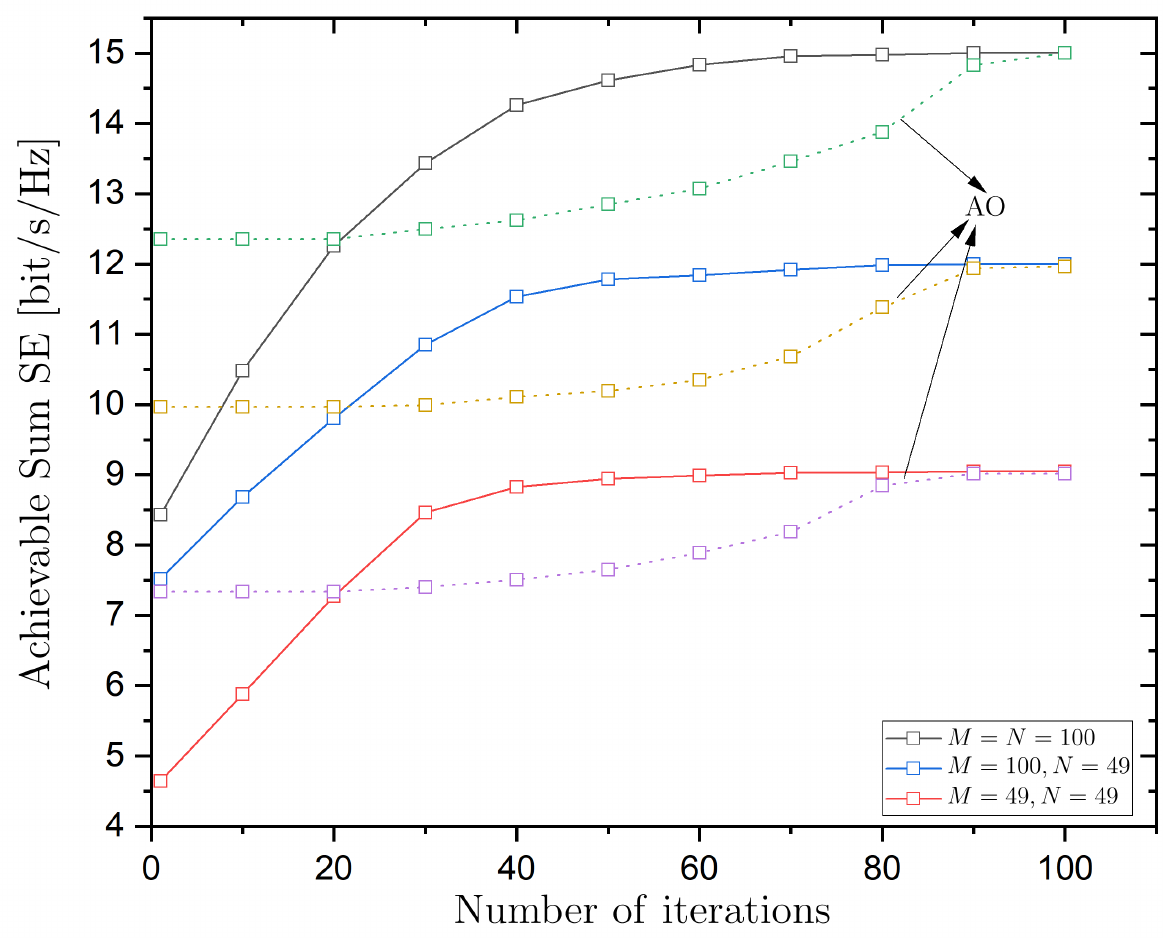}}\qquad
		\subfigure[\textcolor{black}{For   $ 5 $ different randomly generated initial points.} ]{	\includegraphics[width=0.9\linewidth]{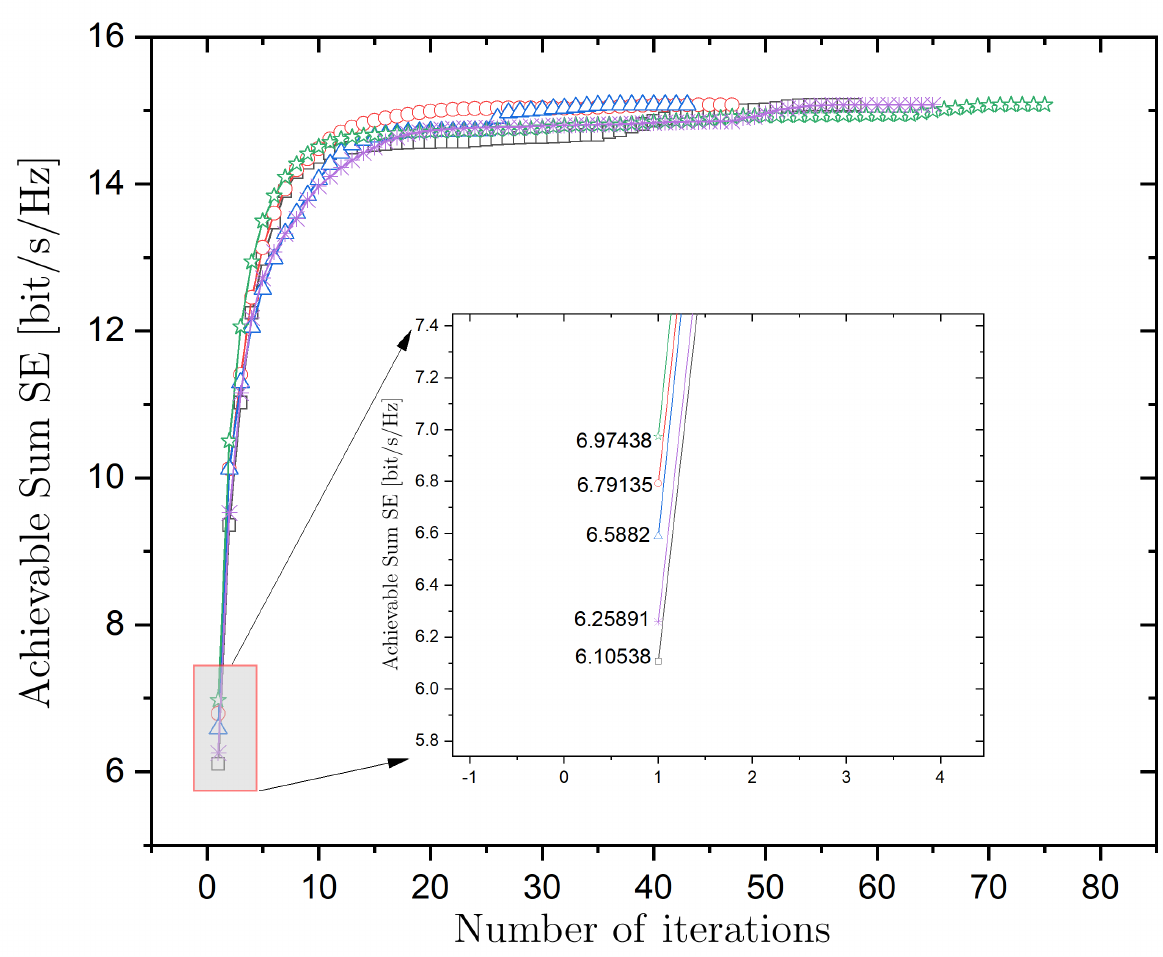}}\\
		\caption{\textcolor{black}{Achievable sum SE of the proposed double-SIM mMIMO architecture  versus the number of iterations}}
		\label{fig2}
	\end{figure}
	\subsection{Channel Estimation Evaluation}
	\textcolor{black}{	The evaluation of the normalized MSE (NMSE) is quite insightful to assess the impact of BSIM and CSIM on the channel estimation. Specifically, we define
		\begin{align}
			\mathrm{NMSE}_{k}&=\frac{\tr(\EE[(\hat{\bc}_{k}-{\bc}_{k})(\hat{\bc}_{k}-{\bc}_{k})^{\H}])}{\tr\left(\EE[{\bc}_{k}{\bc}_{k}^{\H}]\right)}\\
			&=1-\frac{\tr(\bPsi_{k})}{\tr\big(\hat{\bR}_{k}\big)}.\label{nmse1}
	\end{align}}
	
	\textcolor{black}{	Fig. \ref{fig05} depicts the normalized mean square error (NMSE) with respect to the uplink SNR for varying $M$ while $N$ is fixed, and for varying $N$ while $M$ is fixed to study the impact of BCSIM and CSIM, respectively. We observe floors as the SNR increases due to imperfect CSI. Moreover, the variation in the number of elements of CSIM results in a greater decrease in the NMSE. }

	\begin{figure}[!h]
		\begin{center}
			\includegraphics[width=0.9\linewidth]{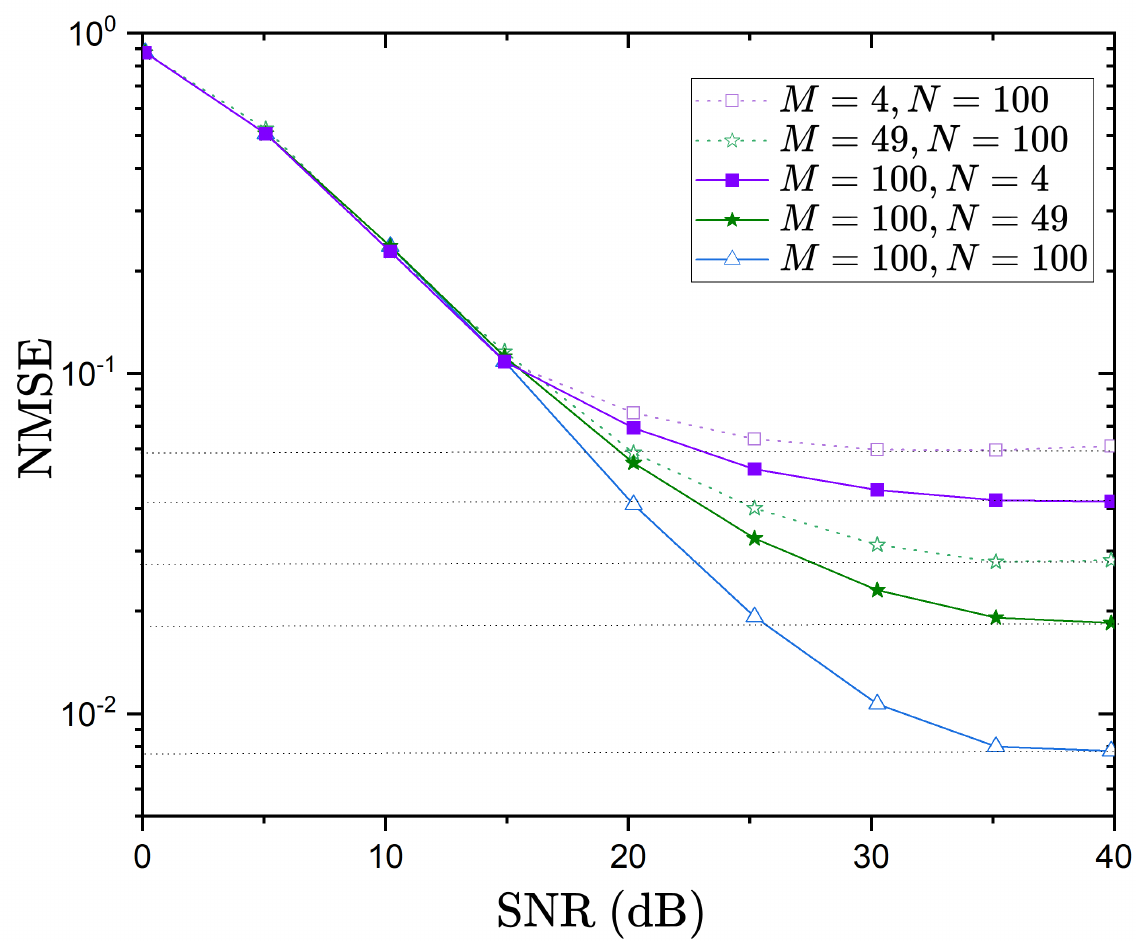}
			\caption{Impact of BCIM and CSIM on the NMSE of UE $ k $ versus the SNR.}
			\label{fig05}
		\end{center}
	\end{figure}
	
	\subsection{Sum SE Performance Evaluation}
	First,  Figs. \ref{fig3}.(a) and \ref{fig3}.(b)  illustrate the sum SE versus the numbers of meta-atoms $ M $  and $ N $ on each metasurface layer of the BSIM and CSIM, respectively.  In both figures, we observe an increase in the sum SE with the sizes of the corresponding surfaces. However, the impact of the CSIM is greater since a double increase in  $N  $ results in $ 170 \% $ improvement compared to $ 200 \% $ improvement when $ M $ becomes almost double (from $ 40 $ to $ 100 $). Also, in  Fig.\ref{fig3}.(b), we observe the impact of the CSIM. In particular, we have shown the SE when the CSIM is absent, and only the direct signal exists. As can be seen, the performance is quite low. 
	
	\begin{figure}%
		\centering
		\subfigure[\textcolor{black}{Versus the number of meta-atoms of the BSIM $ M $ while varying the number of meta-atoms of the CSIM $ N $.}]{	\includegraphics[width=0.9\linewidth]{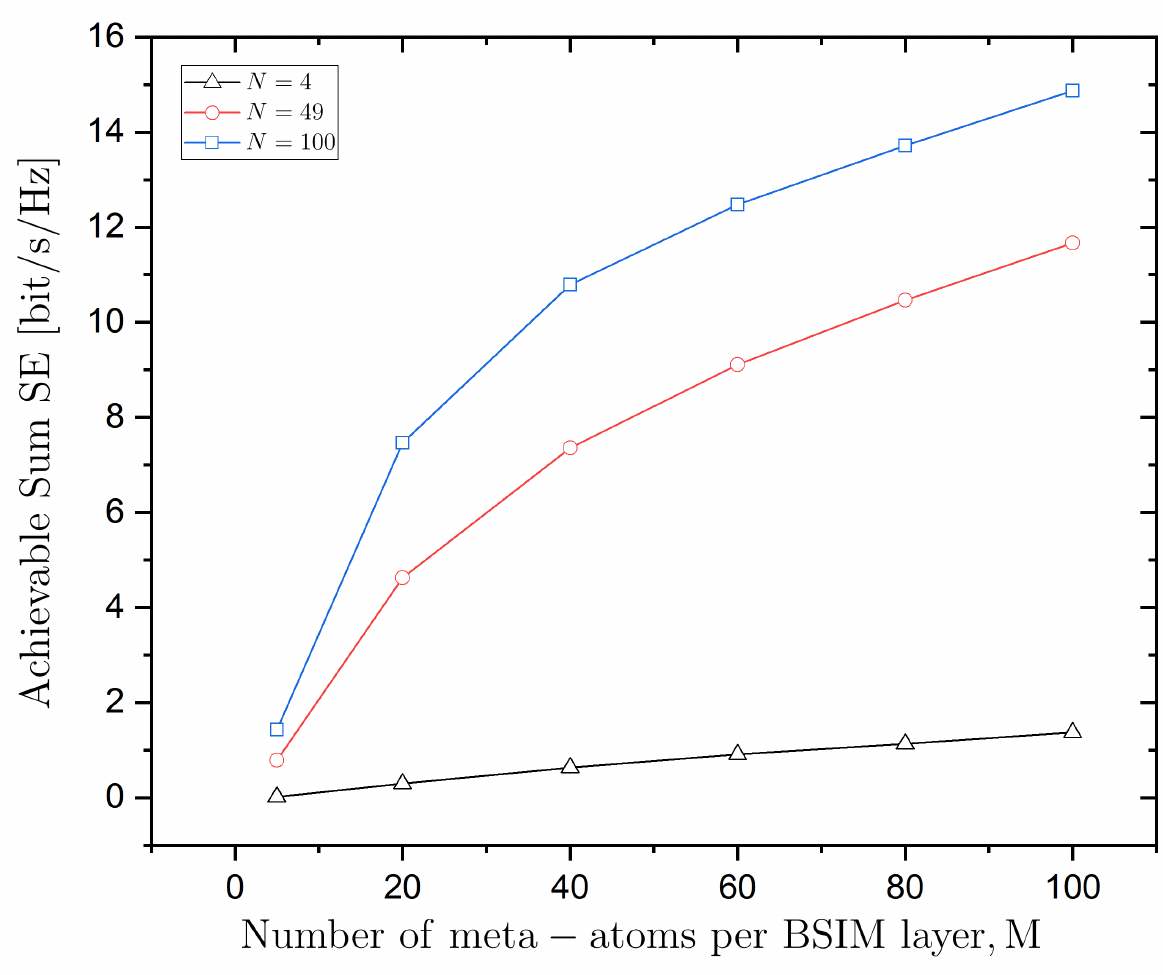}}\qquad
		\subfigure[\textcolor{black}{Versus the number of meta-atoms of the CSIM $ N $ while varying the number of meta-atoms of the BSIM $ M $.}]{	\includegraphics[width=0.9\linewidth]{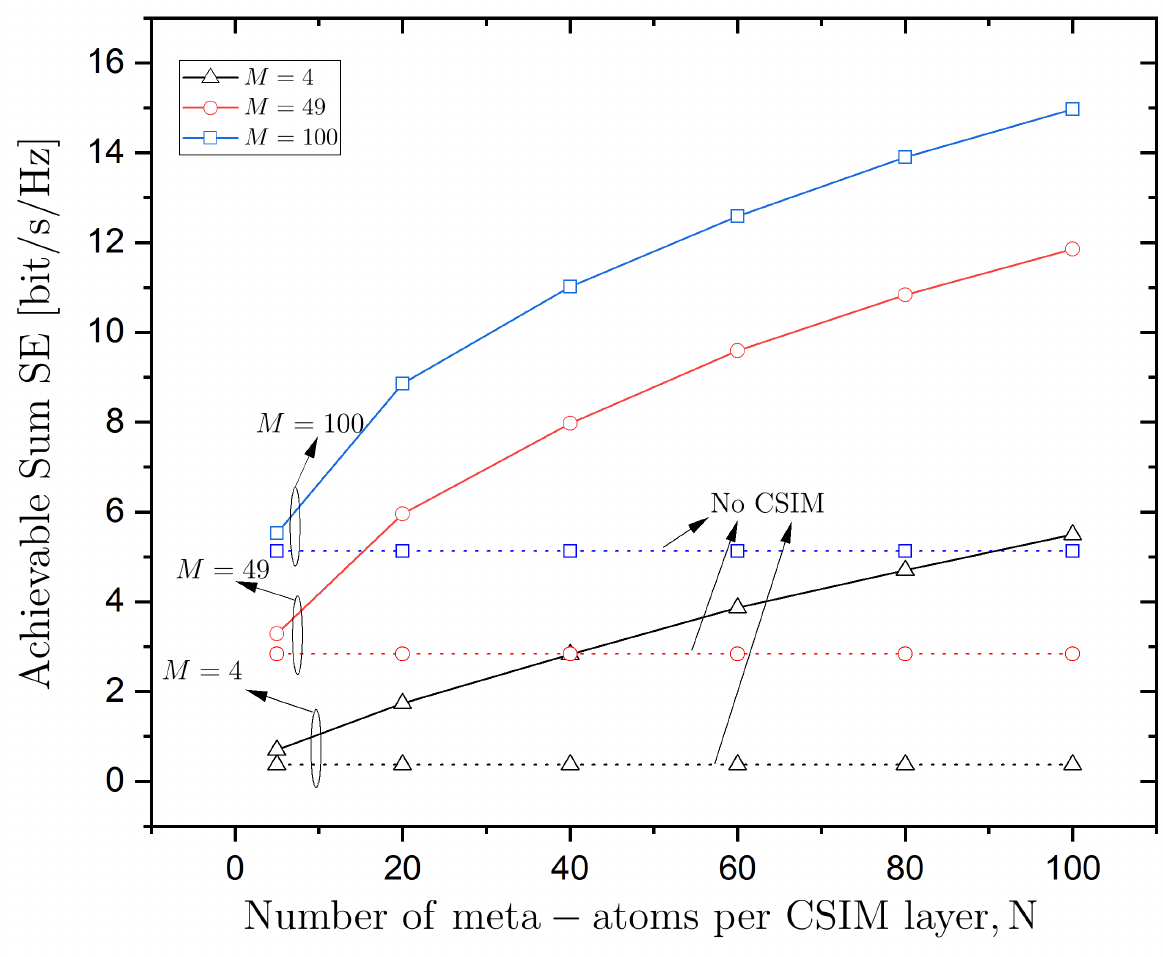}}\\
		\caption{\textcolor{black}{Achievable sum SE of the proposed double-SIM mMIMO architecture.}}
		\label{fig3}
	\end{figure}

	\textcolor{black}{	Figs. \ref{fig4}.(a) and \ref{fig4}.(b)  depict the sum SE versus the numbers of layers $ L $  and $ S $  of the BSIM and CSIM, respectively. In Fig. \ref{fig4}.(a), we observe that the sum SE improves until $ L=4 $ because of the ability of the  BSIM to mitigate the inter-user interference in the EM wave domain. In Fig. \ref{fig4}.(b), it is shown that the sum SE increases with the number of metasurfaces $ S $ due to increasing array and multiplexing gains. In both figures, a significant improvement is observed compared
		to the single-layer BSIM and CSIM, respectively. However, this improvement in the BSIM is greater.}
	
	\begin{figure}%
		\centering
		\subfigure[\textcolor{black}{Versus the number of metasurfaces of the BSIM $ L $ while varying the number of metasurfaces of the CSIM $ S $.}]{	\includegraphics[width=0.9\linewidth]{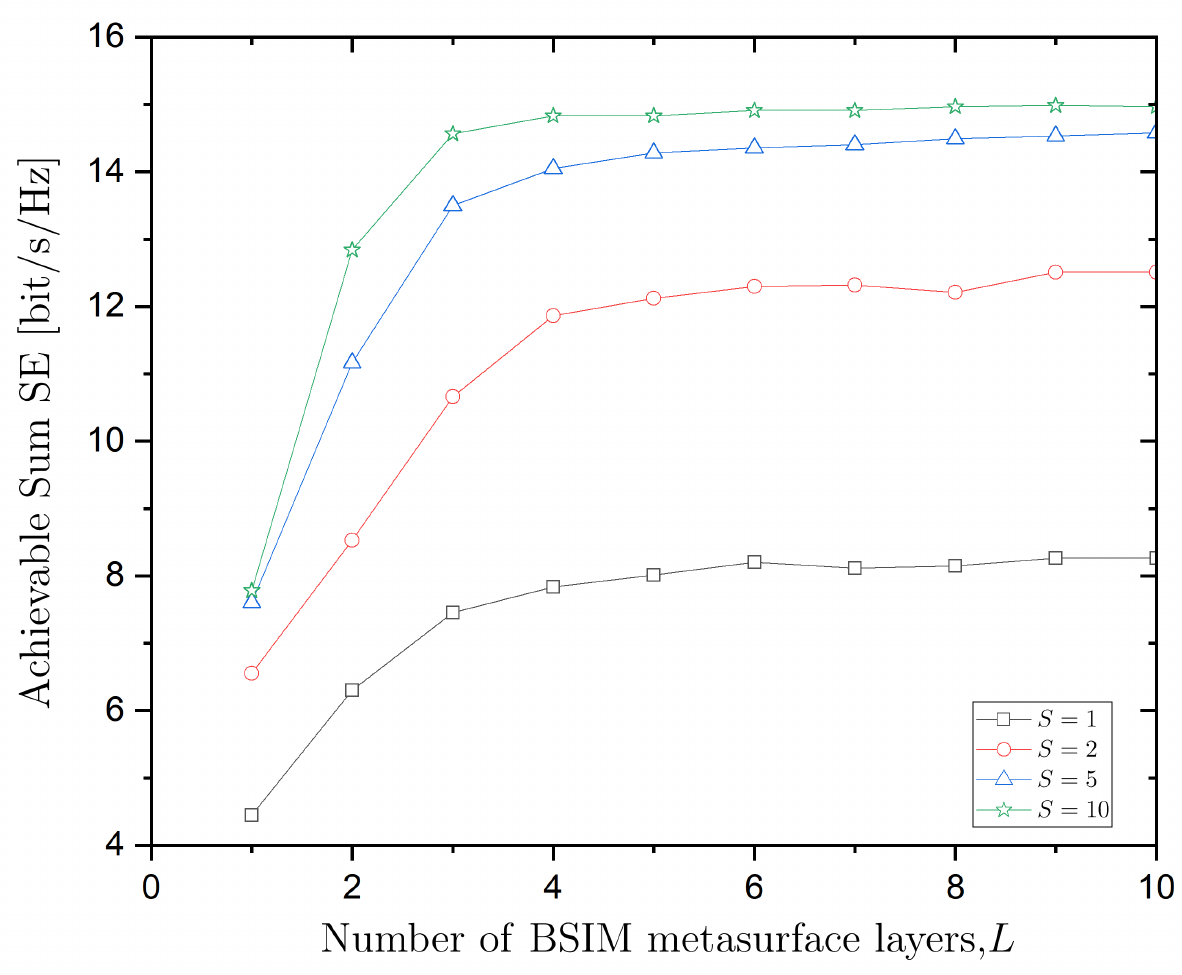}}\qquad
		\subfigure[\textcolor{black}{Versus the number of metasurfaces of the CSIM $ S $ while varying the number of metasurfaces of the BSIM $ L $.}]{	\includegraphics[width=0.9\linewidth]{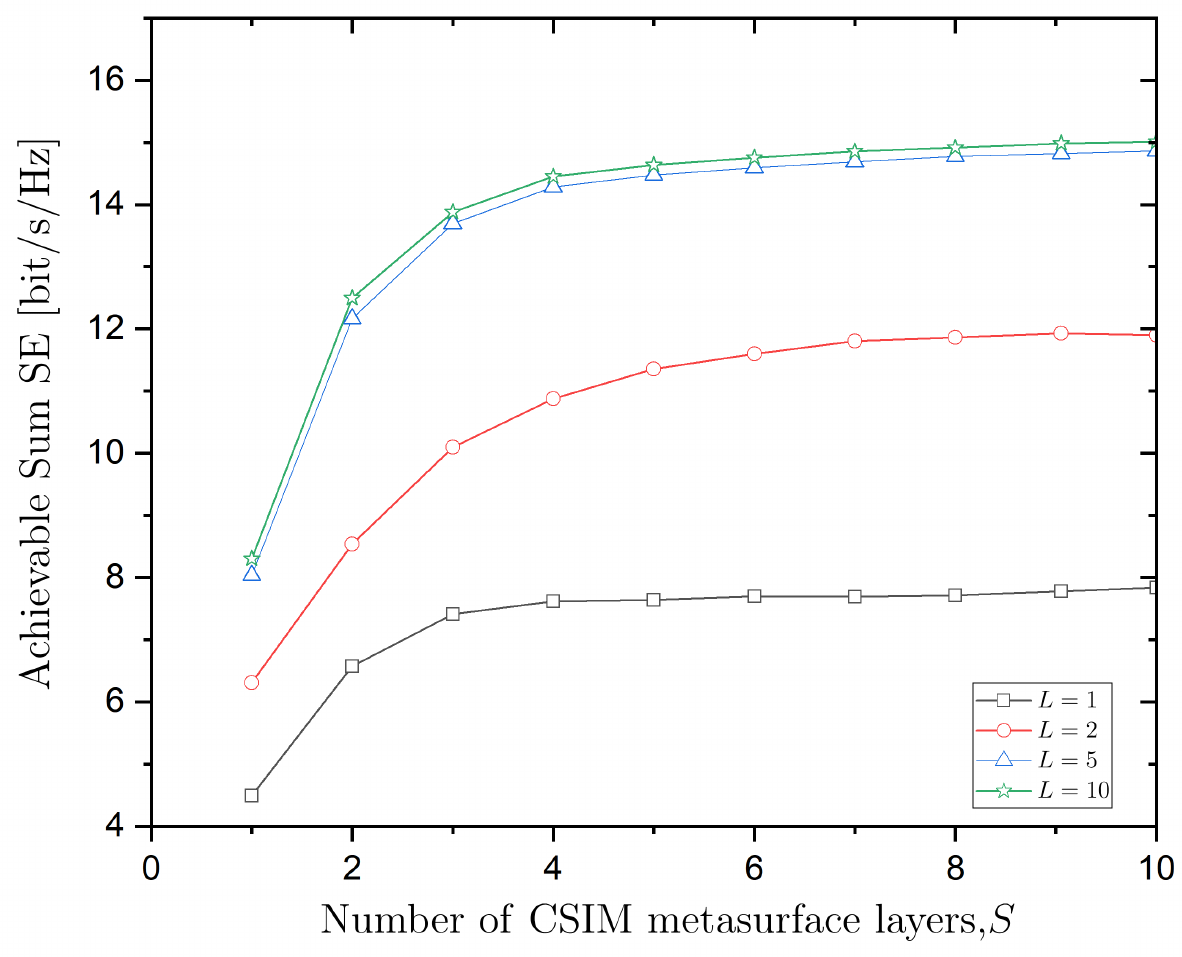}}\\
		\caption{\textcolor{black}{Achievable sum SE of the proposed double-SIM mMIMO architecture.}}
		\label{fig4}
	\end{figure}

	In Fig. \ref{fig7}.(a), we present a comparison  between the BCIM and CSIM-assisted mMIMO system with the case where the BSIM is replaced by a conventional BS that employs a conventional MRC decoder. 	 For the sake of a fair comparison, we assume that the path loss corresponding to user $ k $ is $ d_{k}=\sqrt{H_{\mathrm{BS}^{2}}+(d_{0}(k-1))^{2}} $. The conventional BS requires $M_{\mathrm{BS}}=46  $ antennas to exhibit better performance than the BSIM with $M_{\mathrm{BS}}=32  $ and $ M=49 $ elements. \textcolor{black}{When the number of BSIM elements increase to $ M=100 $, the conventional BS requires $M_{\mathrm{BS}}=81  $ antennas to exhibit better performance.} 
	\textcolor{black}{In Fig. \ref{fig7}.(b), we illustrate the same comparison in the case of the ZF decoder for mitigating the multiuser interference. In this case, the conventional BS requires $M_{\mathrm{BS}}=52  $ antennas to exhibit better performance than the BSIM with $M_{\mathrm{BS}}=32  $ and $ M=49 $ elements. When the number of BSIM elements increases to $ M=100 $, the conventional BS requires $M_{\mathrm{BS}}=89  $ antennas to exhibit better performance.}

	\textcolor{black}{Also, in Fig. \ref{fig7}.(a),  we have plotted the fully wave-based SIM architecture  to show a comparison with the proposed hybrid BSIM architecture in the case of $M_{\mathrm{BS}}=K=8  $. As can be seen, 	 the hybrid BSIM architecture results in a lower performance compared to the fully wave-based BSIM while employing the same number of BS antennas and BSIM elements per surface. To meet the SE  of the fully wave-based BSIM, we have to increase the number of elements of the hybrid BSIM, or else,  the hybrid BSIM achieves a reduction of the metasurface size by making a small performance compromise.	In other words, we can remove partially the digital decoding at the BS while equipping the BS with a hybrid SIM to reduce the size of the fully wave-based SIM. In parallel, the hybrid SIM achieves to deploy fewer antennas at the BS compared to the conventional digital BS, which reduces overall the energy consumption and hardware cost. \textcolor{black}{For example, in Fig. \ref{fig7}.(b), i.e., in the case of ZF, a conventional BS requires $53$ antennas to achieve the same performance with the proposed scheme with $32$ antennas and $49$ elements, while it requires $89$ antennas compared to the proposed architecture, which requires only $32$ antennas and $100$ elements.} \textcolor{black}{The study of the energy consumption in this case could be a very interesting topic for future work.}}

	\begin{figure}%
		\centering
		\subfigure[MRC.]{	\includegraphics[width=0.9\linewidth]{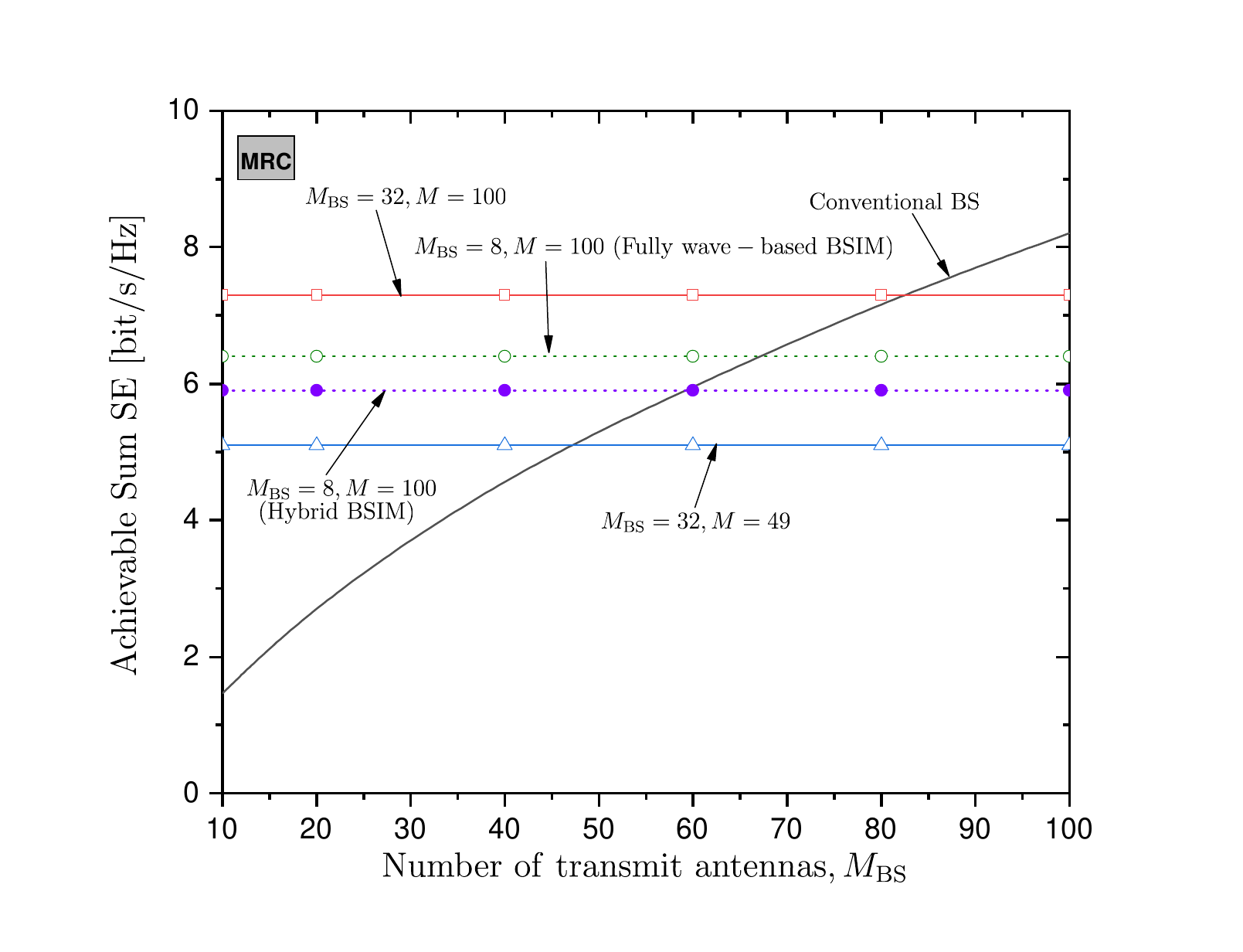}}\qquad
		\subfigure[ZF.]{	\includegraphics[width=0.9\linewidth]{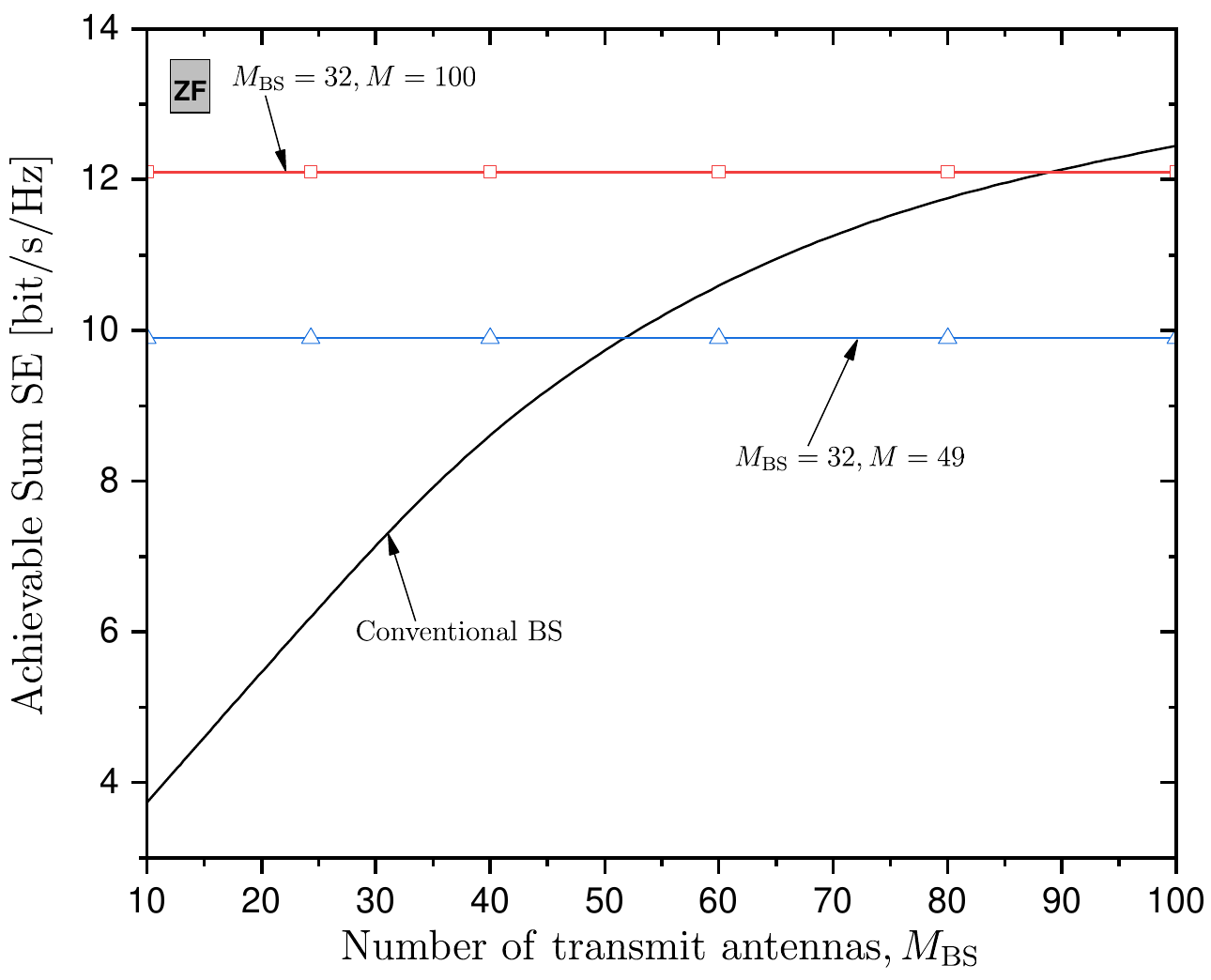}}\\
		\caption{Achievable sum SE of the proposed double-SIM mMIMO architecture versus the number of BS antennas $ M_{\mathrm{BS}} $ while varying the number of receiver metasurface layers $ M $.}
		\label{fig7}
	\end{figure}

	In Fig. \textcolor{black}{\ref{fig5}}, we show the sum SE versus the number of users $ K $.  We observe that the sum SE grows with $ K $ at the beginning but decreases after a certain value of $ K $. The reason behind this observation is that the increase in the beginning comes from the increase of the spatial multiplexing gain with $ K $. Later, the increase of inter-user interference due to increasing $ K $  cannot be mitigated by the BSIM and its wave-based decoding. Based on our setup, the BSIM can mitigate the multiuser interference of  maximum  $ K=5 $ users in the case of MRC. \textcolor{black}{In the same figure, we have plotted the performance with zero-forcing (ZF) (dotted lines) at the BS contrary to MRC (solid lines). It is shown that with ZF, BSIM can serve more users. \textcolor{black}{Moreover,  we observe that an increase in the number of meta-atoms $ M $  of the BSIM does not change much the optimal number of users in the case of MRC, while the change is more significant in the case of ZF.}}  \textcolor{black}{MRC requires more elements than ZF to achieve the same sum SE.  Specifically,	MRC requires 49 elements compared  to 4 elements required  by ZF. The reason is that ZF cancels out the inter-user interference. For a similar reason, the ZF scenario can accommodate a higher number of optimal users, i.e., in the case of ZF, the optimal number of users is 6, while the optimal number of users in the case of MRC is 5.}
	
	\textcolor{black}{	In the same figure, we have added two lines corresponding to the performance of a conventional BS in the cases of MRC and ZF. In the former case, the line corresponds to a conventional BS with $49$ antennas that exhibit the same performance as the proposed architecture with 32 antennas and 49 elements. In the latter case, the  line corresponds to a conventional BS with $53$ antennas  that exhibit the same performance as the proposed architecture with 32 antennas and 49 elements. As expected, in both cases, we observe  that the performance increases with the number of users but at a lower rate as the number of users increases due to multi-user interference. Especially, the case of ZF exhibits better performance. The notable observation compared to the proposed architecture is that the performance of the conventional BS increases with $K$ in both cases. The reason is that the sum SE of the proposed architecture depends on two factors, which are 	the spatial multiplexing gain brought by an increased number of users as well as the residual interference level after applying wave-based beamforming \cite{An2023b}, while the conventional case enjoys fully the spatial multiplexing gain.}

	\begin{figure}[!h]
		\begin{center}
			\includegraphics[width=0.9\linewidth]{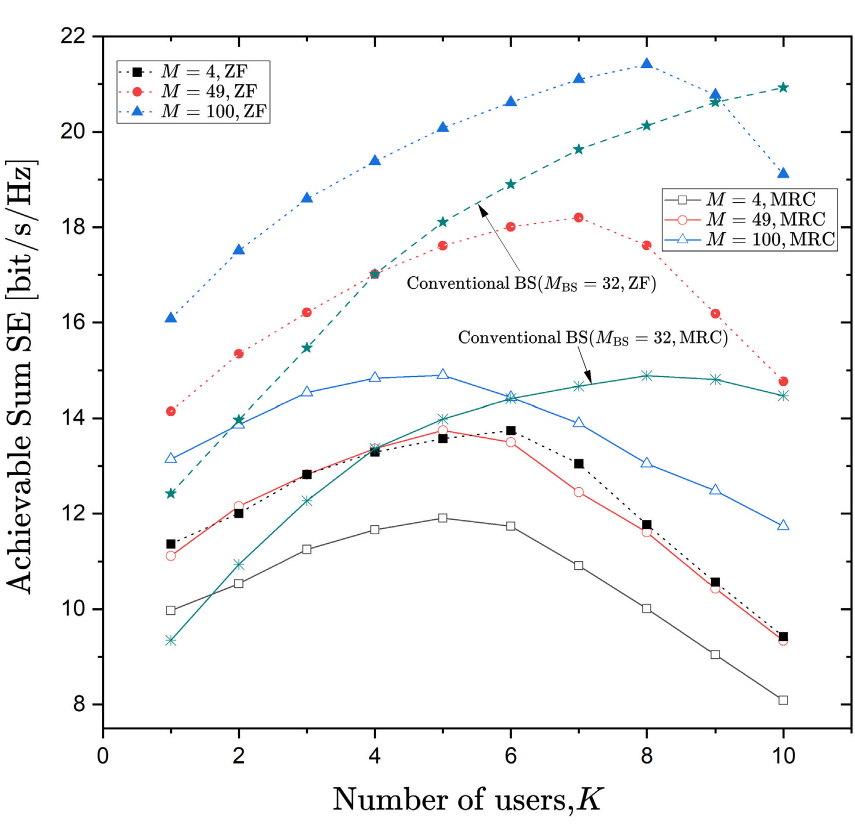}
			\caption{Achievable sum SE of the proposed double-SIM mMIMO architecture versus the number of users $ K $ while varying the number of BSIM meta-atoms $ M $ with MRC and ZF at the BS.}
			\label{fig5}
		\end{center}
	\end{figure}
	
	Fig. \ref{fig6} depicts the sum SE versus the transmit SNR for various cases regarding the values of the phase shifts of the BSIM. 
	The optimal phase shifts design achieves the best performance as expected. The designs with random phase shifts and equal phase shifts follow with lower performance. \textcolor{black}{In the same figure, we have plotted the achievable sum SE in the case of perfect CSI, where it is shown the performance gap compared to imperfect CSI, \textcolor{black}{which is obtained by computing the covariance matrices given by \eqref{Psiexpress} and \eqref{Psiexpress1}.} We observe that the gap increases with increasing SNR.} \textcolor{black}{Moreover, we have simulated the cases of solely a BSIM and solely a CSIM. In particular, in the case of solely a  CSIM, a conventional digital BS has been assumed, while in the case of solely a BSIM, no SIM has been assumed in the intermediate space between the BS and the users. As can be seen, the ``no BSIM'' scenario performs better since a fully digital BS is assumed. However, the ``no CSIM'' setting performs worse since the beamforming gain from the CSIM is missing.} \textcolor{black}{Another comparison, shown in this figure, concerns the performance gap between instantaneous and statistical CSI. The performance loss compared to instantaneous CSI is compromised by  saving significant overhead in the case of statistical CSI. Note  that previous studies on SIMs, which are based on instantaneous CSI changing at each coherence interval, might not be feasible in practice due to inherent large overheads.} \textcolor{black}{In the same figure,  we have depicted the sum SE corresponding to path loss exponents  $2.29$ and 	$1.88$ from \cite[Table VI]{Sang2024}, which come from measurements of RIS channels. This line presents similar performance but the sum SE is higher compared to the scenario, where the path loss exponents are $2.8$ and $ 3.5$, respectively.} In parallel, Monte Carlo simulations verify the analytical results. 
	
	\begin{figure}[!h]
		\begin{center}
			\includegraphics[width=0.9\linewidth]{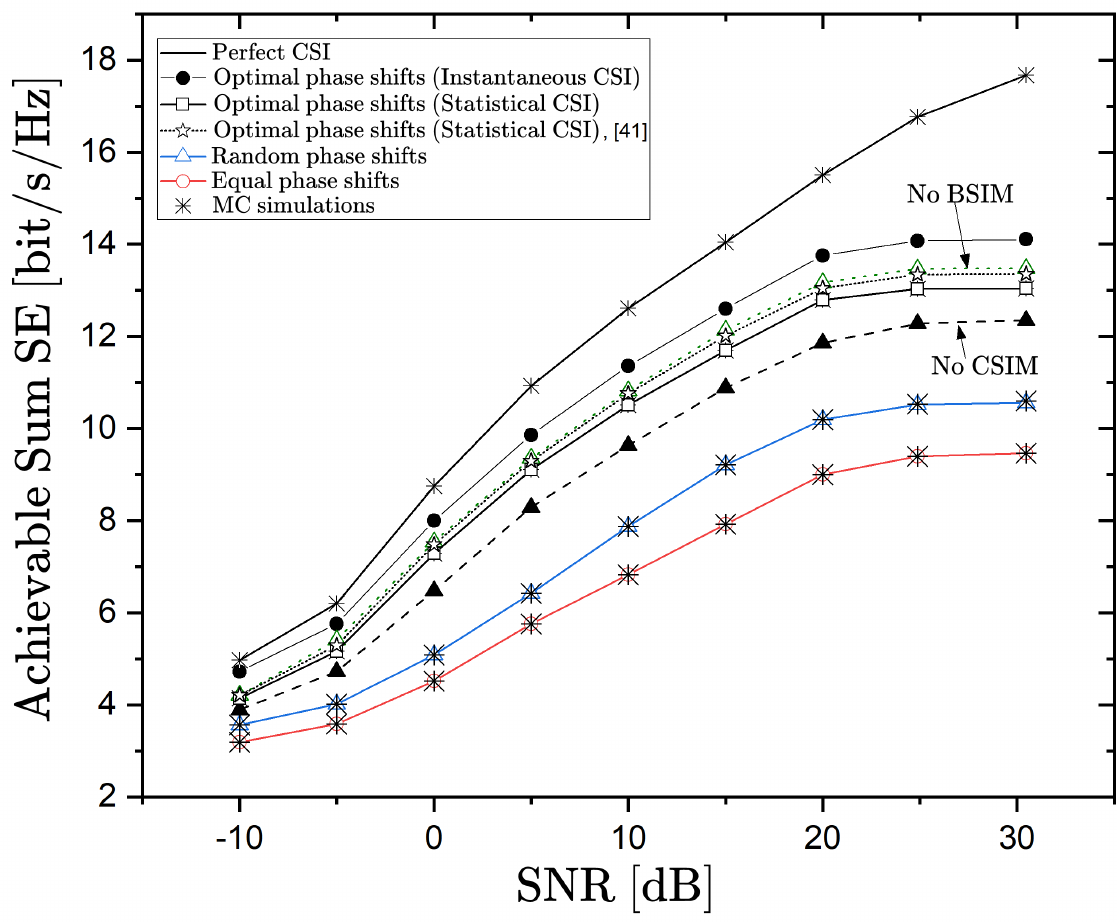}
			\caption{Achievable sum SE of the proposed double-SIM mMIMO architecture versus the SNR for different cases regarding the phase shifts, with perfect CSI, and with instantaneous CSI.}
			\label{fig6}
		\end{center}
	\end{figure}

	\section{Conclusion} \label{Conclusion} 
	In this paper, a novel SIM-assisted mMIMO architecture was proposed. In particular, a SIM enables hybrid wave-based decoding in the uplink, and another SIM is used at the intermediate space between the users and the BS to increase further the ability of the proposed configuration to reshape the surrounding EM environment. We first suggested a channel estimation scheme, which obtains the estimated channel in closed form in a single phase. Then, we derived the uplink sum SE, and optimized both SIMs simultaneously. Simulations showed the superiority of the proposed architecture compared to counterparts, demonstrated the impact of various parameters on the sum SE, and exhibited outperformance against the AO method,  where the SIMs are optimized separately. For example, an increase in the numbers of layers and meta-atoms increases the performance. Especially, the BSIM exhibits interference cancellation. Also, the proposed architecture is superior to the conventional digital MISO counterpart, which means that SIM-based designs could lead to next-generation energy-efficient networks. \textcolor{black}{Notably, an interesting topic for future work could be the mitigation of multiplicative fading caused by a multi-layer CSIM by proposing an active SIM.}
	
	\begin{appendices}
		\section{Proof of Lemma~\ref{LemmaDirectChannel}}\label{Lemma1}
		The LMMSE estimator of $ \bc_{k}=\bW^{1^{\H}}\bP^{H}\bh_{k} $ is derived according to \cite[Ch. 12]{Kay}  as
		\begin{align}
			\hat{\bc}_{k}=\EE\{\br_{k}\textcolor{black}{\bc_{k}}^{\H}\}\left(\EE\{\br_{k}\br_{k}^{\H}\}\right)^{-1}.\label{Cor6}
		\end{align}
		The first term of \eqref{Cor6} is written as
		\begin{align}
			&	\EE\{\br_{k}\bc_{k}^{\H}\}=\EE\big[\big(\bW^{1^{\H}}\bP^{H}\bc_{k}+\frac{\bz_{k}}{ \tau \rho}\big)\bc_{k}^{\H}\bW^{1}\bP\big]\label{Cor00}\\
			&=\bW^{1^{\H}}\bP^{H}\EE\{\bc_{k}\bc_{k}^{\H}\}\bW^{1}\bP=\bW^{1^{\H}}\bP^{H}\bR_{k}\bW^{1}\bP,\label{Cor0}
		\end{align}
		where, \textcolor{black}{in \eqref{Cor00}, we have accounted for the independence between $\bz_{k}$ and $\bc_{k}$}.
		
		The second term is obtained as
		\begin{align}
			\EE\{\br_{k}\br_{k}^{\H}\}&=\bW^{1^{\H}}\bP^{H}\bR_{k}\bW^{1}\bP+\frac{1}{ \tau \rho }\Id_{M}.\label{Cor1}
		\end{align}
		
		By inserting \eqref{Cor0} and \eqref{Cor1} into \eqref{Cor6}, we obtain 
		\begin{align}
			\hat{\bc}_{k}=\hat{\bR}_{k}\bQ_{k}\br_{k},
		\end{align}
		where $ \hat{\bR}_{k}=\bW^{1^{\H}}\bP^{H}\bR_{k}\bW^{1}\bP $ and $ \bQ_{k}\!=\! \left(\!\hat{\bR}_{k}\!+\!\frac{1}{ \tau \rho }\Id_{M}\!\right)^{\!-1}$.
		
		\section{Proof of Theorem~\ref{theorem:ULSINR}}\label{theorem1}
		Regarding the desired signal power part, we have
		\begin{align}
			\EE\left\{\bv_{k}^{\H}\bc_{k}\right\}&=\tr(\EE\{\bc_{k}\hat{\bc}_{k}\})\nn\\
			&=\tr(\EE\{\bc_{k}\hat{\bR}_{k}\bQ_{k} \br_{k}\})\label{desired1}\\
			&=\tr(	\bPsi_{k})\label{desired2},
		\end{align}
		where, in \eqref{desired1}, we have inserted \eqref{received1}, while in the last equation, we have computed the expectation between $   \bc_{k}$ and $  \br_{k} $, and we have used \eqref{Psiexpress}.
		
		In the case of $ I_{k} $, the first term can be written as
		\begin{align}
			&\!\!\EE\left\{|\hat{\bc}_{k}^{\H}\bc_{k}-\EE\left\{\hat{\bc}_{k}^{\H}\bc_{k}\right\}|^{2}\right\}\!=\!
			\EE\big\{ \big|\hat{\bc}_{k}^{\H}\bc_{k}\big|^{2}\big\}\!-\!\big|\EE\left\{\hat{\bc}_{k}^{\H}\bc_{k}\right\}\big|^{2} \label{est2}\\
			&=\EE\big\{ \big| \hat{\bc}_{k}^{\H} \hat{\bc}_{k} +\hat{\bc}_{k}^{\H}\tilde{\bc}_{k}\big|^{2}\big\}-\big|\EE\big\{
			\hat{\bc}_{k}^{\H}\hat{\bc}_{k}\big\}\big|^{2}\label{est3} \\
			&=\EE\big\{ \big|\hat{\bc}_{k}^{\H}\tilde{\bc}_{k}\big|^{2}\big\}+2\Re\{\EE\{\hat{\bc}_{k}^{\H}\hat{\bc}_{k}\hat{\bc}_{k}^{\H}\tilde{\bc}_{k}\}\}\label{est31}
		\end{align}
		where in~\eqref{est3}, we have substituted \eqref{current}.  The second term in~\eqref{est31} goes to zero according to the channel hardening property in mMIMO, which gives with high accuracy $\hat{\bc}_{k}^{\H}\hat{\bc}_{k} \approx\EE\{\hat{\bc}_{k}^{\H}\hat{\bc}_{k}\} $  \cite{Bjoernson2017}.
		Specifically, the second term becomes
		\begin{align}
			\EE\{\hat{\bc}_{k}^{\H}\hat{\bc}_{k}\hat{\bc}_{k}^{\H}\tilde{\bc}_{k}\}\approx\EE\{\hat{\bc}_{k}^{\H}\hat{\bc}_{k}\}\EE\{\hat{\bc}_{k}^{\H}\tilde{\bc}_{k}\}=0.
		\end{align}
		Hence, \eqref{est31} becomes
		\begin{align}
			&\!\!\EE\big\{ \big| \hat{\bc}_{k}^{\H}\bc_{k}-\EE\big\{
			\hat{\bc}_{k}^{\H}\bc_{k}\big\}\big|^{2}\big\}\!\approx \!\EE\big\{|\hat{\bc}_{k}^{\H}\tilde{\bc}_{k}|^{2}\big\} \label{est5}\\
			&=\tr\!\left( \hat{\bR}_{k}\bPsi_{k}\right)-\tr\left( \bPsi_{k}^{2}\right).\label{est4}
		\end{align}
		For the derivation of \eqref{est4}, we apply the approximations $\hat{\bc}_{k}^{\H}\hat{\bc}_{k} \approx\EE\{\hat{\bc}_{k}^{\H}\hat{\bc}_{k}\} $ and $\tilde{\bc}_{k}^{\H}\tilde{\bc}_{k} \approx\EE\{\tilde{\bc}_{k}^{\H}\tilde{\bc}_{k}\} $ due to  channel hardening. Thus, we result in 
		\begin{align}
			\EE\{|\hat{\bc}_{k}^{\H}\tilde{\bc}_{k}|^2\}&=  \tr(\EE\{\hat{\bc}_{k}\hat{\bc}_{k}^{\H}\tilde{\bc}_{k}\tilde{\bc}_{k}^{\H}\})\\
			&\approx \tr(\EE\{\hat{\bc}_{k}\hat{\bc}_{k}^{\H}\}\EE\{\tilde{\bc}_{k}\tilde{\bc}_{k}^{\H}\})\\
			&	= \tr((\hat{\bR}_k-\bPsi_k)\bPsi_k).\label{BU}
		\end{align}
		
		The multiuser interference term in \eqref{int1} becomes
		\begin{align}
			&	\EE\left\{|\hat{\bc}_{k}^{\H}\bc_{i}|^{2}\right\}=\EE\big\{ \big| \hat{\bc}_{k}^{\H}(\hat{\bc}_{i}+\tilde{\bc}_{i})\big|^{2}\big\}\\
			&=\EE\big\{ \big| \hat{\bc}_{k}^{\H}\hat{\bc}_{i}\big|^{2}\big\}+\EE\big\{ \big|\hat{\bc}_{k}^{\H}\tilde{\bc}_{i}\big|^{2}\big\}+2\Re\{\EE\{\hat{\bc}_{k}^{\H}\hat{\bc}_{i}\tilde{\bc}_{i}^{\H}\hat{\bc}_{k}\}\}\label{52}\\
			&=		\tr\!\left(\bPsi_{k}\hat{\bR}_{i} \right),\label{54}
		\end{align}
		where the third term in \eqref{54} is zero because it holds that
		\begin{align}
			2\Re\{\EE\{\hat{\bc}_{k}^{\H}\hat{\bc}_{i}\tilde{\bc}_{i}^{\H}\hat{\bc}_{k}\}\}&\!=\!2\Re\left\{\EE\{\tr\!\left(\!\left(\hat{\bc}_{k}\hat{\bc}_{k}^{\H}\right)\!\!\left(\tilde{\bc}_{i}\hat{\bc}_{i}^{\H}\right)\!\right)\!\right\}\\
			&\!=\!2\Re\left\{\tr\!\left(\EE\{\hat{\bc}_{k}\hat{\bc}_{k}^{\H}\}\EE\{\tilde{\bc}_{i}\hat{\bc}_{i}^{\H}\}\right)\!\right\}\label{521}\\
			&=0\label{522},
		\end{align}
		where, first, we used that $ \hat{\bc}_{k} $ and $ \hat{\bc}_{i} $ are independent, and next that $ \tilde{\bc}_{i} $ and $  \hat{\bc}_{i} $ are uncorrelated.
		
		The last term in \eqref{int1} is written as
		\begin{align}
			\EE\left\{|\hat{\bc}_{k}^{\H}\bn|^{2}\right\}&=\tr(\EE\{\hat{\bc}_{k}\hat{\bc}_{k}^{\H}\})\label{noise1}\\
			&=\tr(\bPsi_{k})\label{noise2},
		\end{align}
		where, in \eqref{noise1}, we have used that $ \hat{\bc}_{k} $ and $ \bn $ are uncorrelated.
		
		The combination of \eqref{est4}, \eqref{54}, and \eqref{noise1} provides the approximate of ${I}_{k}  $, which is $ \tilde{I}_{k} $, and concludes the proof.
		
		\section{Proof of Proposition~\ref{propositionGradient}}\label{proposition1}	
		First, we focus on the derivation of $\nabla_{\bphi_{l}}f(\bphi_{l},\blambda_{s}) $. From \eqref{Maximization1}, we can easily obtain
		\begin{align}
			\nabla_{\bphi_{l}}f(\bphi_{l},\blambda_{s})=\frac{\tau_{c}-\tau}{\tau_{c}\log_{2}(e)}\sum_{k=1}^{K}\frac{\tilde{I}_{k}\nabla_{\bphi_{l}}S_{k}-S_{k}\nabla_{\bphi_{l}}\tilde{I}_{k}}{(1+\gamma_{k})\tilde{I}_{k}}.\label{grad10}
		\end{align}
		For the computation of $ \nabla_{\bphi_{l}}S_{k} $,  we obtain its differential as
		\begin{align}
			d(S_{k})&=d\bigl(\tr(\boldsymbol{\Psi}_{k})^{2}\bigr)\nn\\
			&=2\tr(\boldsymbol{\Psi}_{k})d\tr(\boldsymbol{\Psi}_{k})\nn\\
			&=2\tr(\boldsymbol{\Psi}_{k})\tr(d\boldsymbol{\Psi}_{k}).\label{eq:dSk}
		\end{align}
		Application of  \cite[Eq. (3.35)]{hjorungnes:2011} gives
		\begin{align}
			&d(\boldsymbol{\Psi}_{k})=d(\hat{\mathbf{R}}_{k}\mathbf{Q}_{k}\hat{\mathbf{R}_{k}})\nn\\
			&=d(\hat{\mathbf{R}_{k}})\mathbf{Q}_{k}\hat{\mathbf{R}}_{k}+\hat{\mathbf{R}}_{k}d(\mathbf{Q}_{k})\hat{\mathbf{R}}_{k}+\hat{\mathbf{R}}_{k}\mathbf{Q}_{k}d(\hat{\mathbf{R}}_{k}).\label{eq:dPsik}
		\end{align}
		To this end, we have to derive $ d(\hat{\mathbf{R}_{k}}) $ and $ d(\mathbf{Q}_{k}) $. The former can be written as
		\begin{align}
			d(\hat{\mathbf{R}}_{k})&=\bW^{1^{\H}}d(\bP^{H})\bR_{k}\bW^{1}\bP+\bW^{1^{\H}}\bP^{H}\bR_{k}\bW^{1}d(\bP)\label{diffEstimated}.
		\end{align}
		The latter is derived based on \cite[eqn. (3.40)]{hjorungnes:2011} as
		\begin{align}
			&d(\mathbf{Q}_{k})  =d\bigl(\hat{\mathbf{R}}_{k}+\frac{1}{\tau P}\mathbf{I}_{M}\bigr)^{-1}\nn\\
			&=-\bigl(\hat{\mathbf{R}}_{k}+\frac{1}{\tau P}\mathbf{I}_{M}\bigr)^{-1}d\bigl(\hat{\mathbf{R}}_{k}+\frac{1}{\tau P}\mathbf{I}_{M}\bigr)\bigl(\hat{\mathbf{R}}_{k}+\frac{1}{\tau P}\mathbf{I}_{M}\bigr)^{-1}\nonumber \\
			& =-\mathbf{Q}_{k}d(\hat{\mathbf{R}}_{k})\mathbf{Q}_{k}.\label{eq:dQk}
		\end{align}
		Substitution of \eqref{diffEstimated} and \eqref{eq:dQk} into \eqref{eq:dPsik} allows to obtain 
		$ d(S_{k}) $ from \eqref{eq:dSk} as
		
		\begin{align}
			d(S_{k})	&=2\tr(\boldsymbol{\Psi}_{k})\tr\big(\bA_{l}^{\H}\bR_{k}\bW^{1}\bP\bB_k\bW^{1^{\H}}\bC_{l}^{\H} d(\bPhi^{l^{\H}})\nn\\
			&+\bC_{l}\bB_{k}\bW^{1^{\H}}\bP^{H}\bR_{k}\bW^{1}\bA_{l} d(\bPhi^{l})\big)\label{dsk1}\\
			&=2\tr(\boldsymbol{\Psi}_{k})\left(\diag(\bA_{l}^{\H}\bR_{k}\bW^{1}\bP\bB_k\bW^{1^{\H}}\bC_{l}^{\H})\right)^{\T}d(\bphi^{l^{*}})\nn\\
			&+\left(\diag(\bC_{l}\bB_{k}\bW^{1^{\H}}\bP^{H}\bR_{k}\bW^{1}\bA_{l})\right)^{\T}d(\bphi^{l})\label{dif1},
		\end{align}
		where $\bB_{k}= \mathbf{Q}_{k}\hat{\mathbf{R}}_{k}-\mathbf{Q}_{k}\hat{\mathbf{R}}^{2}_{k}\mathbf{Q}_{k}+\hat{\mathbf{R}}_{k}\mathbf{Q}_{k} $, $ 	\bA_{l}=\bPhi^{L}\bW^{L}\cdots\bPhi^{l+1}\bW^{l+1} $, and $\bC_{l}= \bW^{l}\bPhi^{l-1}\bW^{l-1}\cdots \bPhi^{1} $. In \eqref{dsk1}, we have also used \eqref{TransmitterSIM} by writing its differential as $ d(\bP)=\bA_{l} d(\bPhi_{l})\bC_{l} $.
		
		It can be easily obtained from \eqref{dif1} that
		\begin{align}
			\nabla_{\bphi_{l}}S_{k} &=2\tr(\boldsymbol{\Psi}_{k})\nn\\
			&\times \frac{\partial}{\partial{(\bphi_{l}^{*})}}	
			\Big(\left(\diag(\bA_{l}^{\H}\bR_{k}\bW^{1}\bP\bB_k\bW^{1^{\H}}\bC_{l}^{\H})\right)^{\T}d(\bphi_{l}^{*})\nn\\
			&+\left(\bC_{l}\bB_{k}\bW^{1^{\H}}\bP^{H}\bR_{k}\bW^{1}\bA_{l})\right)^{\T}d(\bphi_{l})\Big)\nn\\
			&=2\tr(\boldsymbol{\Psi}_{k})\diag(\bA_{l}^{\H}\bR_{k}\bW^{1}\bP\bB_k\bW^{1^{\H}}\bC_{l}^{\H})\label{differentialPhi7}.
		\end{align}
		Regarding $\nabla_{\bphi_{l}}\tilde{I}_{k}$, we obtain the differential from \eqref{Den1} as
		\begin{align}
			d(\tilde{I}_{k}) =\tr(\boldsymbol{\Psi}d(\hat{\mathbf{R}}_{k}))+\sum\nolimits _{i=1}^{K}\tr(\hat{\mathbf{R}}_{k}d(\boldsymbol{\Psi}_{i}))+\tr\bigl(\bar{\boldsymbol{\Psi}}_{k}d(\boldsymbol{\Psi}_{k})\bigr),\label{Den2}
		\end{align}
		where $\boldsymbol{\Psi}=\sum\nolimits _{i=1}^{K}\boldsymbol{\Psi}_{i}$, $\bar{\boldsymbol{\Psi}}_{k}=\frac{1}{\rho}\Id-2\boldsymbol{\Psi}_{k}$. Having obtained the differentials in \eqref{Den2} previously and after several algebraic manipulations, we obtain
		\begin{align}
			\nabla_{\bphi_{l}}\tilde{I}_{k}&=\frac{\partial}{\partial\bphi_{l}^{*}}\tilde{I}_{k}\\
			&=\diag\big(\bA_{l}^{\H}\big(\bR_{k}\bW^{1}\bP\big({\bPsi}+\bar{\bPsi}_{k}\bB_k\big)\nn\\
			&+\sum\nolimits _{i=1}^{K}\bR_{i}\bW^{1}\bP\bB_i\big)\bW^{1^{\H}}\bC_{l}^{\H}\big)\label{gradI1}.
		\end{align}
		In the case of $\nabla_{\blambda_{s}}f(\bphi_{l},\blambda_{s}) $, we obtain   similar expressions,  but now $ 	d(\hat{\mathbf{R}}_{k}) $  is derived as
		\begin{align}
			&d(\hat{\mathbf{R}}_{k})=\bW^{1^{\H}}\bP^{H}d(\bR_{k})\bW^{1}\bP\nn\\
			&	=\hat{\beta}_{k}\bW^{1^{\H}}\bP^{H}\tr\big(\bR_{\mathrm{CSIM}}d(\bZ){\bR}_{\mathrm{CSIM}}\bZ^{\H}\nn\\
			&+\bR_{\mathrm{CSIM}}\bZ{\bR}_{\mathrm{CSIM}}d(\bZ^{\H})\big)\bR_{\mathrm{BSIM}}\bW^{1}\bP\label{cov4}\\
			&=\hat{\beta}_{k}\bW^{1^{\H}}\bP^{H}\tr\big(\bF_{s}{\bR}_{\mathrm{CSIM}}\bZ^{\H}\bR_{\mathrm{CSIM}}\bD_{s}d(\bLambda_{s})\nn\\
			&+(\bF_{s}{\bR}_{\mathrm{CSIM}}\bZ^{\H}\bR_{\mathrm{CSIM}}\bD_{s})^{\H}d(\bLambda_{s}^{\H})\big)\bR_{\mathrm{BSIM}}\bW^{1}\bP\label{diffEstimated2},
		\end{align}
		where, in \eqref{52}, we have substituted \eqref{cov3}. We have also denoted
		$ \bD_{s}= \bU^{s-1}\bLambda^{s-1}\cdots\bU^{2}\bLambda^{1}$ and $ \bF_{s}= \bLambda^{S}\bU^{S}\cdots\bLambda^{s+1}\bU^{s+1} $. 
		In \eqref{diffEstimated2}, we have  used \eqref{ReceiverSIM} by writing its differential as $ d(\bZ)=\bD_{s} \bLambda_{s}\bF_{s} $.

		Substitution of \eqref{diffEstimated2} and \eqref{eq:dQk} into \eqref{eq:dPsik} allows to obtain 
		$ d(S_{k}) $ from \eqref{eq:dSk} as
		\begin{align}
			d(S_{k})	
			&=2\hat{\beta}_{k}\tr(\boldsymbol{\Psi}_{k})\tr(\bW^{1^{\H}}\bP^{H}\bR_{\mathrm{BSIM}}\bW^{1}\bP\bB_{k})\nn\\
			&\times\left(\diag(\bF_{s}{\bR}_{\mathrm{CSIM}}\bZ^{\H}\bR_{\mathrm{CSIM}}\bD_{s})^{\H}\right)^{\T}d(\blambda_{s}^{*})\nn\\
			&+\left(\diag(\bF_{s}{\bR}_{\mathrm{CSIM}}\bZ^{\H}\bR_{\mathrm{CSIM}}\bD_{s})\right)^{\T}d(\blambda_{s})\label{dif5}.
		\end{align}
		
		Thus, we obtain 
		\begin{align}
			\nabla_{\blambda_{s}}S_{k}
			&=2\hat{\beta}_{k}\tr(\boldsymbol{\Psi}_{k})\tr(\bW^{1^{\H}}\bP^{H}\bR_{\mathrm{BSIM}}\bW^{1}\bP\bB_{k})\nn\\
			&\times\diag(\bF_{s}{\bR}_{\mathrm{CSIM}}\bZ^{\H}\bR_{\mathrm{CSIM}}\bD_{s})^{\H}.
		\end{align}
		Similar to \eqref{gradI1}, we obtain from \eqref{Den2} and \eqref{diffEstimated2} and after some algebraic manipulations
		\begin{align}
			\nabla_{\blambda_{s}}\tilde{I}_{k}&=
			2\tr\big(\bW^{1^{\H}}\bP^{H}\bR_{\mathrm{BSIM}}\bW^{1}\bP\nn\\
			&\times(\sum\nolimits _{i=1}^{K}\hat{\beta}_{i}\hat{\mathbf{R}}_{k}\bB_{i}+\hat{\beta}_{k}(\bar{\boldsymbol{\Psi}}_{k}\bB_{k}+\boldsymbol{\Psi})) \big)&\nn\\
			& \times\diag(\bF_{s}{\bR}_{\mathrm{CSIM}}\bZ^{\H}\bR_{\mathrm{CSIM}}\bD_{s})^{\H},
		\end{align}
		which concludes the proof.
	\end{appendices}

	\bibliographystyle{IEEEtran}
	
	\bibliography{IEEEabrv,bibl}

\begin{thebibliography}{10}
\providecommand{\url}[1]{#1}
\csname url@samestyle\endcsname
\providecommand{\newblock}{\relax}
\providecommand{\bibinfo}[2]{#2}
\providecommand{\BIBentrySTDinterwordspacing}{\spaceskip=0pt\relax}
\providecommand{\BIBentryALTinterwordstretchfactor}{4}
\providecommand{\BIBentryALTinterwordspacing}{\spaceskip=\fontdimen2\font plus
\BIBentryALTinterwordstretchfactor\fontdimen3\font minus
  \fontdimen4\font\relax}
\providecommand{\BIBforeignlanguage}[2]{{%
\expandafter\ifx\csname l@#1\endcsname\relax
\typeout{** WARNING: IEEEtran.bst: No hyphenation pattern has been}%
\typeout{** loaded for the language `#1'. Using the pattern for}%
\typeout{** the default language instead.}%
\else
\language=\csname l@#1\endcsname
\fi
#2}}
\providecommand{\BIBdecl}{\relax}
\BIBdecl

\bibitem{Boccardi2014}
F.~Boccardi \emph{et~al.}, ``Five disruptive technology directions for {5G},''
  \emph{IEEE Commun. Mag.}, vol.~52, no.~2, pp. 74--80, 2014.

\bibitem{Zhang2020b}
J.~Zhang \emph{et~al.}, ``Prospective multiple antenna technologies for beyond
  {5G},'' \emph{IEEE J. Sel. Areas Commun.}, vol.~38, no.~8, pp. 1637--1660,
  2020.

\bibitem{Sohrabi2016}
F.~Sohrabi and W.~Yu, ``Hybrid digital and analog beamforming design for
  large-scale antenna arrays,'' \emph{IEEE J. Sel. Top. Signal Proc.}, vol.~10,
  no.~3, pp. 501--513, 2016.

\bibitem{Rappaport2015}
T.~S. Rappaport \emph{et~al.}, ``Wideband millimeter-wave propagation
  measurements and channel models for future wireless communication system
  design,'' \emph{IEEE Trans. Commun.}, vol.~63, no.~9, pp. 3029--3056, 2015.

\bibitem{Wu2019}
Q.~Wu and R.~Zhang, ``Intelligent reflecting surface enhanced wireless network
  via joint active and passive beamforming,'' \emph{IEEE Trans. Wireless
  Commun.}, vol.~18, no.~11, pp. 5394--5409, 2019.

\bibitem{Basar2019}
E.~Basar \emph{et~al.}, ``Wireless communications through reconfigurable
  intelligent surfaces,'' \emph{IEEE Access}, vol.~7, pp. 116\,753--116\,773,
  2019.

\bibitem{Bjoernson2020}
E.~{Bj{\"o}rnson} and L.~{Sanguinetti}, ``Rayleigh fading modeling and channel
  hardening for reconfigurable intelligent surfaces,'' \emph{IEEE Wireless
  Commun. Lett.}, vol.~10, no.~4, pp. 830--834, 2021.

\bibitem{Papazafeiropoulos2021}
A.~Papazafeiropoulos \emph{et~al.}, ``Intelligent reflecting surface-assisted
  {MU-MISO} systems with imperfect hardware: {Channel} estimation and
  beamforming design,'' \emph{IEEE Trans. Wireless Commun.}, vol.~21, no.~3,
  pp. 2077--2092, 2021.

\bibitem{Papazafeiropoulos2022}
A.~Papazafeiropoulos, ``Ergodic capacity of {IRS}-assisted {MIMO} systems with
  correlation and practical phase-shift modeling,'' \emph{IEEE Wireless Commun.
  Lett.}, vol.~11, no.~2, pp. 421--425, 2022.

\bibitem{Huang2019}
C.~Huang \emph{et~al.}, ``Reconfigurable intelligent surfaces for energy
  efficiency in wireless communication,'' \emph{IEEE Transa. Wireless Commun.},
  vol.~18, no.~8, pp. 4157--4170, 2019.

\bibitem{DiRenzo2020}
M.~Di~Renzo \emph{et~al.}, ``Smart radio environments empowered by
  reconfigurable intelligent surfaces: {How} it works, state of research, and
  the road ahead,'' \emph{IEEE J. Sel. Areas Commun.}, vol.~38, no.~11, pp.
  2450--2525, 2020.

\bibitem{Kammoun2020}
Q.~U.~A. {Nadeem} \emph{et~al.}, ``Asymptotic max-min {SINR} analysis of
  reconfigurable intelligent surface assisted {MISO} systems,'' \emph{IEEE
  Trans. Wireless Commun.}, vol.~19, no.~12, pp. 7748--7764, 2020.

\bibitem{Yang2020b}
Y.~Yang \emph{et~al.}, ``Intelligent reflecting surface meets {OFDM: Protocol}
  design and rate maximization,'' \emph{IEEE Trans. Commun.}, vol.~68, no.~7,
  pp. 4522--4535, 2020.

\bibitem{Pan2020}
C.~{Pan} \emph{et~al.}, ``Multicell {MIMO} communications relying on
  intelligent reflecting surfaces,'' \emph{IEEE Trans. Wireless Commun.},
  vol.~19, no.~8, pp. 5218--5233, 2020.

\bibitem{Han2019}
Y.~Han \emph{et~al.}, ``Large intelligent surface-assisted wireless
  communication exploiting statistical {CSI},'' \emph{IEEE Trans. Veh. Tech.},
  vol.~68, no.~8, pp. 8238--8242, 2019.

\bibitem{Zhang2021a}
J.~Zhang \emph{et~al.}, ``Large system achievable rate analysis of
  {RIS}-assisted {MIMO} wireless communication with statistical {CSIT},''
  \emph{IEEE Trans. Wireless Commun.}, vol.~20, no.~9, pp. 5572--5585, 2021.

\bibitem{Zhao2020}
M.-M. Zhao \emph{et~al.}, ``Intelligent reflecting surface enhanced wireless
  networks: {Two}-timescale beamforming optimization,'' \emph{IEEE Trans.
  Wireless Commun.}, vol.~20, no.~1, pp. 2--17, 2020.

\bibitem{Papazafeiropoulos2021a}
A.~Papazafeiropoulos \emph{et~al.}, ``Coverage probability of distributed {IRS}
  systems under spatially correlated channels,'' \emph{IEEE Wireless Commun.
  Lett.}, vol.~10, no.~8, pp. 1722--1726, 2021.

\bibitem{Gan2022}
X.~Gan \emph{et~al.}, ``Multiple {RISs} assisted cell-free networks with
  two-timescale {CSI: Performance} analysis and system design,'' \emph{IEEE
  Trans. Commun.}, vol.~70, no.~11, pp. 7696--7710, 2022.

\bibitem{Zhi2023}
K.~Zhi \emph{et~al.}, ``Two-timescale design for reconfigurable intelligent
  surface-aided massive {MIMO} systems with imperfect {CSI},'' \emph{IEEE
  Trans. Inf. Theory}, vol.~69, no.~5, pp. 3001--3033, 2023.

\bibitem{Papazafeiropoulos2023}
A.~Papazafeiropoulos \emph{et~al.}, ``Achievable rate of a {STAR-RIS} assisted
  massive {MIMO} system under spatially-correlated channels,'' \emph{IEEE
  Trans. Wireless Commun.}, pp. 1--1, 2023.

\bibitem{Papazafeiropoulos2023a}
A.~Papazafeiropoulos, P.~Kourtessis, and S.~Chatzinotas, ``{Max-Min SINR}
  analysis of {STAR-RIS} assisted massive {MIMO} systems with hardware
  impairments,'' \emph{IEEE Trans. Wireless Commun.}, pp. 1--1, 2023.

\bibitem{An2023}
J.~An \emph{et~al.}, ``Stacked intelligent metasurfaces for efficient
  holographic {MIMO} communications in {6G},'' \emph{IEEE J. Sel. Areas
  Commun.}, 2023.

\bibitem{Lin2018}
X.~Lin \emph{et~al.}, ``All-optical machine learning using diffractive deep
  neural networks,'' \emph{Science}, vol. 361, no. 6406, pp. 1004--1008, 2018.

\bibitem{Liu2022}
C.~Liu \emph{et~al.}, ``A programmable diffractive deep neural network based on
  a digital-coding metasurface array,'' \emph{Nature Electronics}, vol.~5,
  no.~2, pp. 113--122, 2022.

\bibitem{An2023b}
J.~An \emph{et~al.}, ``Stacked intelligent metasurfaces for multiuser downlink
  beamforming in the wave domain,'' \emph{arXiv preprint arXiv:2309.02687},
  2023.

\bibitem{Nadeem2023}
Q.-U.-A. Nadeem, J.~An, and A.~Chaaban, ``Hybrid digital-wave domain channel
  estimator for stacked intelligent metasurface enabled multi-user {MISO}
  systems,'' \emph{arXiv preprint arXiv:2309.16204}, 2023.

\bibitem{Abeywickrama2020}
S.~Abeywickrama \emph{et~al.}, ``Intelligent reflecting surface: {Practical}
  phase shift model and beamforming optimization,'' \emph{IEEE Trans. Commun.},
  vol.~68, no.~9, pp. 5849--5863, 2020.

\bibitem{Wang2024}
Z.~Wang \emph{et~al.}, ``Multi-user {ISAC} through stacked intelligent
  metasurfaces: {New} algorithms and experiments,'' \emph{arXiv preprint
  arXiv:2405.01104}, 2024.

\bibitem{Nadeem2020}
Q.~{Nadeem} \emph{et~al.}, ``Intelligent reflecting surface-assisted multi-user
  {MISO Communication: Channel} estimation and beamforming design,'' \emph{IEEE
  Open J. Commun. Soc.}, vol.~1, pp. 661--680, 2020.

\bibitem{Deshpande2022}
N.~V. Deshpande \emph{et~al.}, ``Spatially-correlated irs-aided multiuser {FD
  mMIMO} systems: {Analysis} and optimization,'' \emph{IEEE Trans. Commun.},
  vol.~70, no.~6, pp. 3879--3896, 2022.

\bibitem{Wu2021}
C.~Wu \emph{et~al.}, ``Channel estimation for {STAR-RIS}-aided wireless
  communication,'' \emph{IEEE Commun. Letters}, vol.~26, no.~3, pp. 652--656,
  2021.

\bibitem{Zheng2022}
B.~Zheng \emph{et~al.}, ``A survey on channel estimation and practical passive
  beamforming design for intelligent reflecting surface aided wireless
  communications,'' \emph{IEEE Commun. Sur. \& Tut.}, vol.~24, no.~2, pp.
  1035--1071, 2022.

\bibitem{Bjoernson2017}
E.~Bj{\"o}rnson \emph{et~al.}, ``Massive {MIMO} networks: Spectral, energy, and
  hardware efficiency,'' \emph{Foundations and Trends{\textregistered} in
  Signal Processing}, vol.~11, no. 3-4, pp. 154--655, 2017.

\bibitem{Hoydis2013}
J.~Hoydis, S.~ten Brink, and M.~Debbah, ``Massive {MIMO} in the {UL/DL} of
  cellular networks: How many antennas do we need?'' \emph{IEEE J. Select.
  Areas Commun.}, vol.~31, no.~2, pp. 160--171, 2013.

\bibitem{Papazafeiropoulos2015a}
A.~K. Papazafeiropoulos and T.~Ratnarajah, ``Deterministic equivalent
  performance analysis of time-varying massive {MIMO} systems,'' \emph{IEEE
  Trans. Wireless Commun.}, vol.~14, no.~10, pp. 5795--5809, 2015.

\bibitem{Papazafeiropoulos2016}
A.~K. Papazafeiropoulos, ``Impact of general channel aging conditions on the
  downlink performance of massive {MIMO},'' \emph{IEEE Trans. Veh. Tech.},
  vol.~66, no.~2, pp. 1428--1442, Feb 2017.

\bibitem{Zhang2020a}
S.~Zhang and R.~Zhang, ``Capacity characterization for intelligent reflecting
  surface aided {MIMO} communication,'' \emph{IEEE J. Sel. Areas Commun.},
  vol.~38, no.~8, pp. 1823--1838.

\bibitem{Perovic2021}
N.~S. Perovi{\'c} \emph{et~al.}, ``Achievable rate optimization for {MIMO}
  systems with reconfigurable intelligent surfaces,'' \emph{IEEE Trans.
  Wireless Commun.}, vol.~20, no.~6, pp. 3865--3882, 2021.

\bibitem{Bertsekas1999}
D.~Bertsekas, \emph{Nonlinear Programming}, 2nd~ed., M.~A. Scientific, Ed.,
  1999.

\bibitem{Sang2024}
J.~Sang \emph{et~al.}, ``Multi-scenario broadband channel measurement and
  modeling for {Sub-6 GHz RIS}-assisted wireless communication systems,''
  \emph{IEEE Trans. Wireless Commun.}, vol.~23, no.~6, pp. 6312--6329, 2024.

\bibitem{Kay}
S.~M. Kay, \emph{Fundamentals of statistical signal processing: Estimation
  theory}.\hskip 1em plus 0.5em minus 0.4em\relax Upper Saddle River: Prentice
  Hall PTR, 1993.

\bibitem{hjorungnes:2011}
A.~Hj{\o}rungnes, \emph{Complex-Valued Matrix Derivatives: With Applications in
  Signal Processing and Communications}.\hskip 1em plus 0.5em minus 0.4em\relax
  Cambridge University Press, 2011.

\end{thebibliography}
\end{document}